%% file: StabilityAnalysis_rsj2014.tex
\documentclass[twocolumn]{autart}

\input{crenv}
\graphicspath{{figuresTR/}}

\begin{document}

\begin{frontmatter}
\title{Stability Analysis and Design of \\
a Network of Event-based Systems}
\thanks[footnoteinfo]{This work was supported by the Swedish Research Council, VINNOVA (The Swedish Governmental Agency for Innovation Systems), the Swedish Foundation for Strategic Research, the Knut and Alice Wallenberg Foundation and the EU project Hycon$2$.}
\author{Chithrupa Ramesh}\ead{cramesh@kth.se},    
\author{Henrik Sandberg}\ead{hsan@kth.se},  
\author{Karl H. Johansson}\ead{kallej@kth.se}
\address{ACCESS Linnaeus Centre, KTH Royal Institute of Technology, Electrical Engineering, Stockholm, Sweden}


\begin{abstract}
We consider a network of event-based systems that use a shared wireless medium to communicate with their respective controllers. These systems use a contention resolution mechanism to arbitrate access to the shared network. We identify sufficient conditions for Lyapunov mean square stability of each control system in the network, and design event-based policies that guarantee it. Our stability analysis is based on a Markov model that removes the network-induced correlation between the states of the control systems in the network. Analyzing the stability of this Markov model remains a challenge, as the event-triggering policy renders the estimation error non-Gaussian. Hence, we identify an auxiliary system that furnishes an upper bound for the variance of the system states. Using the stability analysis, we design policies, such as the constant-probability policy, for adapting the event-triggering thresholds to the delay in accessing the network. Realistic wireless networked control examples illustrate the applicability of the presented approach.
\end{abstract}
\end{frontmatter}
%
%
%
%


\begin{abstract}
We consider a network of event-based systems that use a shared wireless medium to communicate with their respective controllers. These systems use a contention resolution mechanism to arbitrate access to the shared network. We identify sufficient conditions for Lyapunov mean square stability of each control system in the network, and design event-based policies that guarantee it. Our stability analysis is based on a Markov model that removes the network-induced correlation between the states of the control systems in the network. Analyzing the stability of this Markov model remains a challenge, as the event-triggering policy renders the estimation error non-Gaussian. Hence, we identify an auxiliary system that furnishes an upper bound for the variance of the system states. Using the stability analysis, we design policies, such as the constant-probability policy, for adapting the event-triggering thresholds to the delay in accessing the network. Realistic wireless networked control examples illustrate the applicability of the presented approach.
\end{abstract}

\section{Introduction} \label{S:Intro}

\subsection{Motivation}
A wireless networked control system comprises of many plants that communicate with their respective controllers over a shared wireless network. The control systems may use the shared network for sensing, as depicted in Fig.~\ref{Fig:NCSoverview}, or for actuation, or both. Using a wireless network brings many benefits such as mobility, ease of adding sensors and reduced wiring costs. However, using a common medium to communicate data from multiple control systems can result in congestion, which degrades control performance. This can be mitigated by reducing the number of transmissions from each control system. Event-based systems provide a means of accomplishing this, by transmitting only select events in place of periodic samples from the plant~\cite{Otanez2002,Astrom2002}.

\begin{figure*}[tb]
\begin{center}
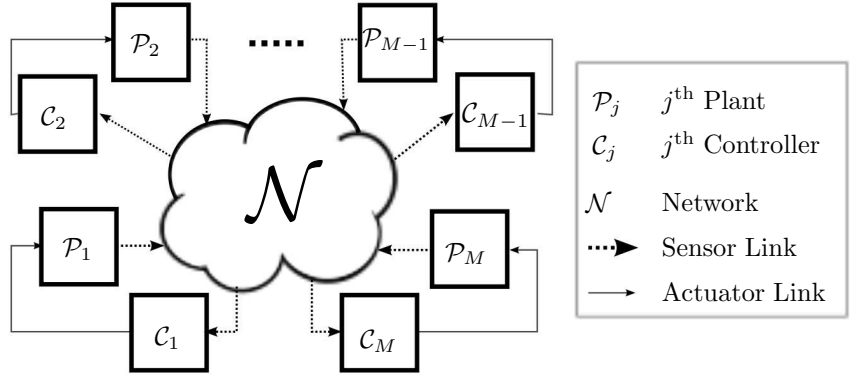
\caption{$M$ control systems comprising of a plant ($\mathcal{P}_j$) and a controller ($\mathcal{C}_j$) each, use a shared network ($\mathcal{N}$) for communication between their respective sensors and controllers. The controllers and actuators communicate over dedicated point to point links. We seek to design an event-triggering policy that results in stability for each control system, and for the network itself. } \label{Fig:NCSoverview}
\end{center}
\end{figure*}

From the network perspective, periodic packet arrivals from multiple sources can be easily scheduled. However, plant-generated event arrivals cannot be anticipated and consequently, cannot be scheduled a priori. Also, more than one event may arise simultaneously, resulting in a conflict for access to the shared medium. Thus, a Contention Resolution Mechanism (CRM), commonly used in wireless networking, is required to resolve such conflicts~\cite{Rom1990}. The CRM is a distributed and non-coordinated protocol, and hence, sometimes results in collisions, wherein all the data packets involved are lost. Now, an essential question for wireless networked control systems is how should event-based systems be designed to compensate for such losses, and yet provide a control system guarantee?

The event-triggering policies considered here are designed to detect a level crossing in the plant output~\cite{Rabi2006}. If the physical medium causes packet losses, then altering the triggering level to permit more frequent transmissions improves the packet reception rate~\cite{Rabi2009}. However, this strategy may not work when packets are lost due to collisions. This is because increasing the number of events may increase the number of collisions. Thus, the triggering levels may have to be altered to reduce the number of transmissions, so as to alleviate congestion in the network. This is the same principle used in congestion control in TCP/IP or in the backoff mechanism in Carrier-Sense Multiple Access (CSMA) protocols. However, will such a policy lead to stability of the networked control system? In other words, how should the levels be selected, to ensure stability of the network \emph{and} stability of the control system? Answering this question is the main objective of this paper.

\subsection{Main Contributions}
There are two main contributions in this paper. Our first contribution is to identify stability conditions for a network of event-based systems. To analyze stability of this network, we use the network-interaction model proposed in~\cite{Ramesh2011b}. Here, Bianchi's assumption~\cite{Bianchi2000} is used to decouple interaction between the various loops, resulting in a steady state Markov model. A statistical description of the system evolution through the states of the Markov chain is not analytically tractable, and hence, we identify an upper bound to describe the system using majorization theory. We obtain sufficient conditions for Lyapunov mean square stability by analyzing the resulting upper bound, and find that this notion of stability is achievable, if the probability of increasing delay is suitably restricted.

Our second contribution is to use the above stability analysis to design event-triggering policies that guarantee stability. We introduce a constant-law policy, where the event probabilities are mandated to remain constant, with increasing delay. We derive conditions for Lyapunov mean square stability for this policy, and present a design algorithm that guarantees it for a network of control systems using this policy. Hence, the paper is constructive in delivering an explicit policy guaranteeing network and closed-loop stability under suitable assumptions.

\subsection{Related Work}

Event-based systems were proposed as a means to reduce congestion in Networked Control Systems (NCS)~\cite{Astrom1999,Yook2002,Otanez2002}. Early work showed that the same control performance can be achieved using fewer samples with event-based systems, for a single system~\cite{Tomovic1966,Astrom1999}. Various event-triggering policies have been proposed for different problem formulations, both stochastic~\cite{Rabi2006,Henningsson2008} and deterministic~\cite{Tabuada2007,Heemels2008}. The event-triggering policies considered in this paper use the estimation error to decide when to transmit. Different variations of policies that use the innovations or estimation error can be found for networked estimation~\cite{Xia2013,Han2013} and networked control~\cite{Molin2009,Demirel2013}. Measurement policies based on the innovations have been suggested much earlier, notably in the encoder design problem for data-rate limited channels~\cite{Borkar1997}.

The multiple access problem for event-based systems has not received as much attention. Much of the work focussing on the design of event-based systems for a shared network~\cite{Wang2011,Molin2012} does not explicitly deal with the problem of multiple access. Others use protocols such as the CAN bus for wired networks~\cite{Anta2009a}, or dynamic real-time scheduling for multiple tasks on a single processor~\cite{Tabuada2007}. These protocols are not well-suited to wireless networks~\cite{Akyildiz1999,Gummalla2000}. There have been some attempts to analyze a network of event-based systems with random access~\cite{Cervin2008,Rabi2009,Henningsson2010}, albeit with simplifying assumptions such as independent packet losses, or by ignoring collisions. More recently, event-based systems which use Aloha and Slotted Aloha have been analyzed~\cite{Blind2011a}, but with an event-triggering policy that is not adapted to the network. In this paper, we use the Markov chain from~\cite{Ramesh2011b} to model the interactions between the event-triggering policy and the CRM. A similar Markov chain has been used, but to model only the event-triggering policy, in~\cite{Demirel2013,Xia2013}.

The problem of level selection after a packet loss was introduced in~\cite{Rabi2009}, where the authors evaluated the control cost of level triggering subject to i.i.d packet losses. Stochastic stability of event-based systems with i.i.d intervals between arrivals have been studied in~\cite{Antunes2011,Lemmon2011}. However, event arrivals in a contention-based network are not i.i.d, and the event arrivals considered in this paper exhibit a dependence on the delay since the last transmission. The notion of stability that we use in this analysis has been used in~\cite{Gupta2010}, to analyze i.i.d erasures, with a provision to extend to Markov models, in NCSs.

\begin{figure*}[tb]
\begin{center}
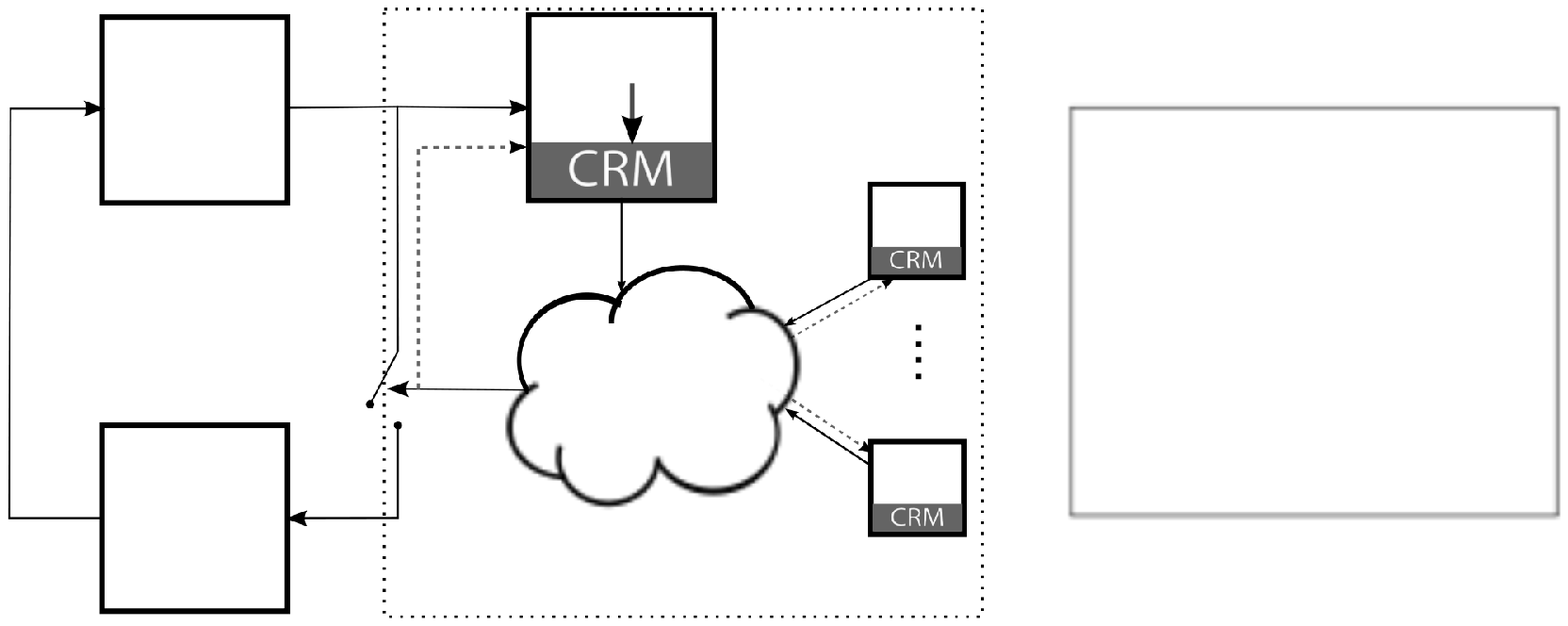
\caption{An overview of a multiple access network ($\smash{\mathcal{N}}$) of plants ($\smash{\mathcal{P}_{j}}$), state-based schedulers ($\smash{\mathcal{S}_{j}}$) and controllers ($\smash{\mathcal{C}_{j}}$), for $\smash{j \in \{1,\dots,M\}}$, using a CRM to access the shared network.} \label{Fig:SystemModel_CRM}
\end{center}
\end{figure*}

\subsection{Outline}
The rest of this paper is organized as follows. The problem formulation, along with a Markov chain representation, is presented in Section~\ref{S:ProbForm}. The main results on sufficient conditions for Lyapunov mean square stability are presented in Section~\ref{S:Results}, and three design laws are presented in Section~\ref{S:SchedDesign}. Some examples and conclusions follow in Sections~\ref{S:Sim} and \ref{S:Conclusions}, respectively.

\section{Problem Formulation} \label{S:ProbForm}

We consider a network of $M$ event-based systems, shown in Fig.~\ref{Fig:SystemModel_CRM}. We first describe a model for each event-based system in the network, indexed by $\smash{j \in \{1,\dots,M\}}$, and then present a model for the interaction of the $M$ systems.

\subsection{Closed-loop System Model}
The network on the sensor link can be modelled from the perspective of a single control system, as illustrated in Fig.~\ref{Fig:DualPred_CRM}. We describe each block in this model below. When the context is clear, we skip the system index $j$.

\noindent \textbf{Plant}: The plant $\smash{\mathcal{P}_{j}}$ has state dynamics given by
\begin{equation}
\label{Eq:StateSpace} x^{j}_{k+1} = A_j x^{j}_k + B_j u^{j}_k + w^{j}_k \; ,
\end{equation}
where $\smash{x^{j}_k \in \mathbb{R}^{n}}$, $\smash{u^{j}_k \in \mathbb{R}^{m}}$ and the initial state $\smash{x^{j}_0}$ and the process noise $\smash{w^{j}_k}$ are i.i.d. zero-mean Gaussians with covariance matrices $\smash{R^{j}_0}$ and $\smash{R^{j}_w}$, respectively. They are independent and uncorrelated to each other and to the initial states and process noises of other plants in the network. This discrete time model is defined with respect to a sampling period $T$ for each plant, and the sampling instants are generated by a synchronized network clock.

\noindent \textbf{Scheduler}: A local scheduler $\smash{\mathcal{S}_j}$, situated in the sensor node, executes the event-triggering policy. The event indicator is denoted $\smash{\gamma^{j}_k \in \{0,1\}}$, with $\smash{\gamma^{j}_k=1}$ in the case of an event. The event-triggering policy uses the innovations process to determine $\smash{\gamma^{j}_k}$, as given by
\begin{equation}
\gamma^{j}_k = \begin{cases}
1, & ||x^{j}_k - \hat{x}^{s,j}_{^{k|\tau_{k-1}}}|| > \Delta^{j}_d, \\
0, & \textrm{otherwise}.
\end{cases}
\label{Eq:InnoSched}
\end{equation}
Here, $\smash{\hat{x}^{s,j}_{^{k|\tau_{k-1}}} = A_j \hat{x}^{c,j}_{^{k-1|k-1}} + B_j u^{j}_{k-1}}$ and $\smash{\hat{x}^{c,j}_{^{k-1|k-1}}}$ denotes the estimate at the controller, defined in (\ref{Eq:EstimateO}) below. Furthermore, $\smash{\tau^{j}_k}$ is the time index of the last received packet, given by $\smash{\tau^{j}_k = \{ \max \{n,-1\} : \delta^{j}_n = 1, n \le k \}}$. Also, $\smash{\Delta^{j}_d > 0}$ is the event threshold, and it may vary with the delay $\smash{d = d^{j}_k \triangleq k-\tau^{j}_k}$. The parameters $\smash{\tau^{j}_k}$ and $\smash{d^{j}_k}$ are illustrated in Fig.~\ref{Fig:TimeLines}. To realize the above event-triggering policy, the observer and controller must be replicated within the scheduler, and an explicit acknowledgement (ACK) of a successful transmission is required.

\noindent \textbf{Network}: The network $\mathcal{N}$ generates exogenous traffic, as is indicated by $\smash{n^{j}_k \in \{0,1\}}$. It takes a value $1$ when a network source generates an event, and $0$ otherwise. The network traffic is stochastic, and hence, $\smash{n^{j}_k \in \{0,1\}}$ is not required to be i.i.d.

\begin{figure*}[!t]
\begin{center}
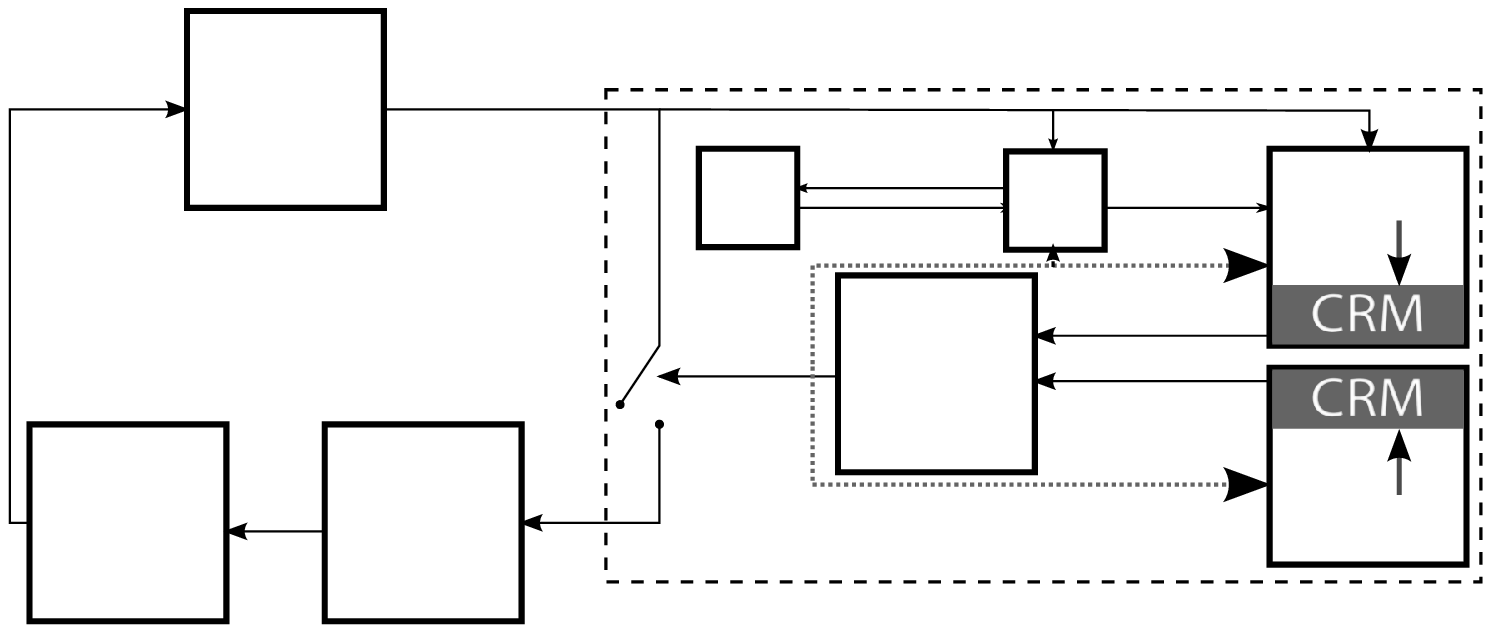
\caption{A model of the control system from the perspective of a single NCS in the network. The other control loops in the network are abstracted by the network traffic block ($\mathcal{N}$). The resolution block ($\mathcal{R}$) maps the CRM output $\alpha$ to the channel access indicator $\delta$. A copy of the observer ($\mathcal{O}$) and controller ($\mathcal{C}$) are required at the scheduler. } \label{Fig:DualPred_CRM}
\end{center}
\end{figure*}

\noindent \textbf{CRM}: The CRM resolves contention between simultaneous channel access requests. For simplicity, we assume that the network uses $p$-persistent CSMA with either no retransmissions or multiple retransmissions. We describe the CRM without retransmissions here, and explain how our model extends to the multiple retransmissions case in Section~\ref{SS:JointAnal}. The CRM output is denoted $\smash{\alpha^{j}_k \in \{0,1\}}$, and we have
\begin{equation} \label{Eq:alphaCRM}
\Pr(\alpha^{j}_k = 1|\gamma^{j}_k=1) = p_{_{\alpha}}
\end{equation}
where $p_{_{\alpha}}$ denotes the persistence probability of the CRM. Thus, with probability $q_{_{\alpha}} = 1-p_{_{\alpha}}$, some events are suppressed by the CRM and not permitted to access the medium. Similarly, $\alpha^{N,j}_k$ is the CRM output for the rest of the network, and $\Pr(\alpha^{N,j}_k = 1|n^{j}_k=1) = p_{_{\alpha}}$.

The resolution block ($\mathcal{R}$) maps the CRM outputs $\alpha^{j}_k$ and $\alpha^{N,j}_k$ to the channel access indicator $\delta^{j}_k$, as given by
\begin{equation} \label{Eq:deltaCRM}
\delta^{j}_k = \alpha^{j}_k (1-\alpha^{N,j}_k)
\end{equation}
where $(\delta^{j}_k=1)$ indicates that a successful transmission of the event has occurred. This is possible only when the CRM permits a transmission \emph{and} none of the other nodes attempt to transmit.

\noindent \textbf{Observer} $(\mathcal{O}_{j})$: The input to the observer is the received measurement signal $y^{j}_k = \delta^{j}_k x^{j}_k$. The observer generates the estimate $\hat{x}^{c,j}_{^{k|k}}$ as given by
\begin{equation} \label{Eq:EstimateO}
\hat{x}^{c,j}_{^{k|k}} = (1-\delta^{j}_k) (A_j \hat{x}^{c,j}_{^{k-1|k-1}} + B_j u^{j}_{k-1}) + \delta^{j}_k x^{j}_k \; ,
\end{equation}
where the estimate for $\delta^{j}_k = 0$ is the model-based prediction from the last received data packet at time $\tau^{j}_k$. The estimation error is defined as $\tilde{x}^{c,j}_{^{k|k}} \triangleq x^{j}_k - \hat{x}^{c,j}_{^{k|k}}$, and $P^{j}_{^{k|k}} = \E[\tilde{x}^{c,j}_{^{k|k}} (\tilde{x}^{c,j}_{^{k|k}})^{T}]$ is the covariance of the estimation error. We denote the variance as $\tr\{P^{j}_{^{k|k}}\}$, where $\tr$ is the trace operator.

\noindent \textbf{Controller} $(\mathcal{C}_j)$: The controller generates the signal $u^{j}_k$ as given by
\begin{equation}
\label{Eq:Controller}
u^{j}_k = - L_j \hat{x}^{c,j}_{^{k|k}} \; ,
\end{equation}
where $L_j$ is the controller gain chosen to minimize a control cost, such as an infinite horizon Linear Quadratic Gaussian (LQG) cost function. 

We are interested in investigating mean square boundedness of the plant state in steady state, or equivalently Lyapunov mean square stability. It is defined below for a control system in the above network. We skip the index $j$ as the definition is applicable for each control system.
\begin{definition} \textbf{\upshape{(Lyapunov Mean Square Stability~\cite{Kozin1969})}} \label{Def:MSS}
A state is said to possess Lyapunov mean square stability if given $\zeta > 0$, there exists $\xi(\zeta) > 0$ such that $|x_0| < \xi$ implies
\begin{equation} \label{Eq:LMSS}
\limsup_{k \rightarrow \infty} \E[x_k^T x_k] \le \zeta \; .
\end{equation}
\end{definition}

The Certainty Equivalence Principle has been shown to hold in the architecture described in (\ref{Eq:StateSpace})--(\ref{Eq:Controller}) in~\cite{Ramesh2013}. Thus, we can translate the stability property in Definition~\ref{Def:MSS} from the state to the estimation error, as shown below.
\begin{lemma} \label{Lemma:Lyapunov mean square stability_CE}
For a control system in the network given by (\ref{Eq:StateSpace})--(\ref{Eq:Controller}), there exists a constant $\varsigma$, with $0 < \varsigma < \zeta$, such that (\ref{Eq:LMSS}) is equivalent to
\begin{equation} \label{Eq:LMSS_CE}
\limsup_{k \rightarrow \infty} \tr\{\E[P_{^{k|k}}]\} \le \varsigma \; .
\end{equation}
\end{lemma}
\begin{proof}
The estimate at the controller in (\ref{Eq:EstimateO}) can be rewritten as
\begin{align}
\hat{x}^c_{^{k|k}} &= (A-BL) \hat{x}^c_{^{k-1|k-1}} + \delta_k (A \tilde{x}^c_{k-1|k-1} + w_{k-1}) \; . \label{Eq:EstAlt}
\end{align}
Since $\hat{x}^c_{^{k-1|k-1}}$ is the minimum mean square error estimate~\cite{Ramesh2013}, we have $\E[x_k^T x_k] = \tr\{(A-BL) \E[\hat{x}^{c}_{^{k-1|k-1}} (\hat{x}^{c}_{^{k-1|k-1}})^{T}] (A-BL)^T\} + \tr\{\E[P_{^{k|k}}]\}$, which must be bounded in steady state for stability, as per Definition~\ref{Def:MSS}. Certainty equivalence implies that the control law ensures mean square boundedness of the estimate $\hat{x}^c_{^{k-1|k-1}}$ in (\ref{Eq:EstAlt}). Hence, the stability condition depends only on the estimation error, so $x_k$ possesses Lyapunov mean square stability iff $\limsup_{k \rightarrow \infty} \tr\{\E[P_{^{k|k}}]\} \le \varsigma$. \hfill \qed
\end{proof}
In the rest of the paper, we identify sufficient conditions that guarantee Lyapunov mean square stability, in the sense of (\ref{Eq:LMSS_CE}), for the states of each of the $M$ control systems described above. Furthermore, we seek a design procedure for selecting the event thresholds, $\Delta_d$, so as to guarantee Lyapunov mean square stability for the overall network of systems.

\begin{figure}[tb]
\centering
\def\svgwidth{8cm}
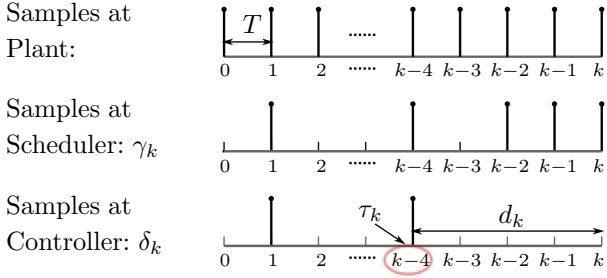
\caption{An illustration of the delay since the last received packet, $\smash{d_k}$, and the index of the last received packet, $\smash{\tau_k}$. Events ($\smash{\gamma_k = 1}$) are chosen from the samples, and only some events are successfully transmitted ($\smash{\delta_k = 1}$). } \label{Fig:TimeLines}
\end{figure}

\subsection{Network Interaction Model} \label{SS:JointAnal}
We have defined a model and a notion of stability for each control system. Next, we model the interactions in the network of $M$ control systems, and define a notion of stability for the entire network. We first present an example to motivate the need for such a model.

\begin{figure*}[tb]
\centering
\includegraphics*[scale=0.45,viewport=60 20 1050 300]{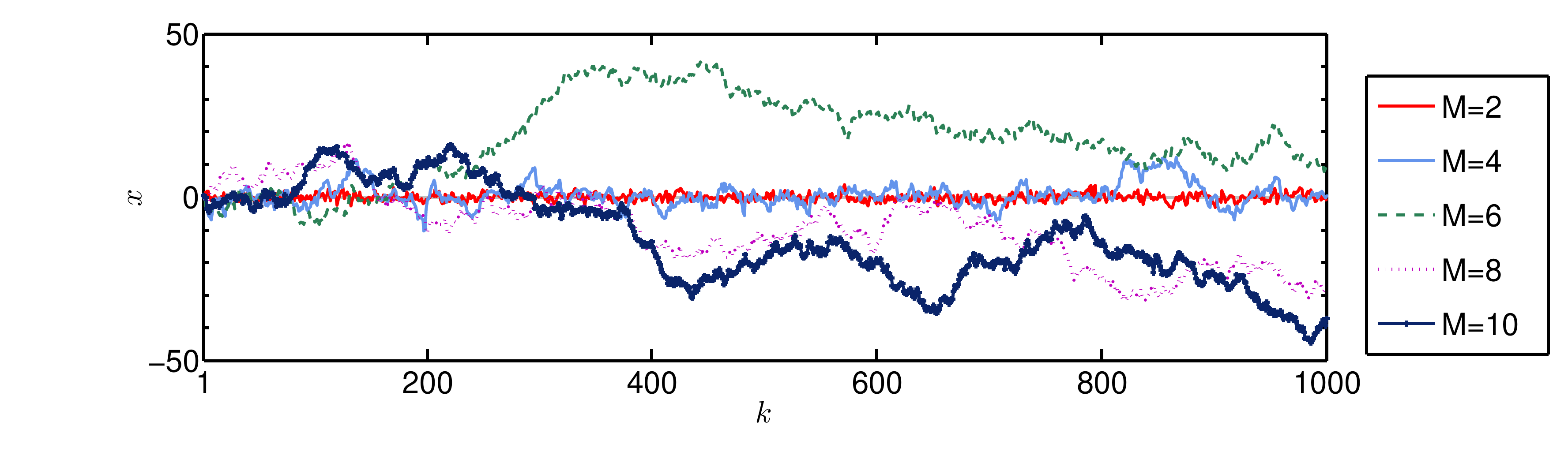}
\caption{A comparison of the trace of the plant state $x$ for identical plants in different sized networks; the event-triggering design appears to result in stability for small sized networks only. }
\label{Fig:MultiX}
\end{figure*}

\begin{example} \label{Ex:LMSSvisualizer}
We consider scenarios corresponding to $M \in \{2,4,6,8,10\}$ event-based systems. The rest of the parameters, chosen identically for all the systems in the network, are $A=1$, $B=1$, $\Sigma_w = 1$ and $\Delta_d=0.25$, $\forall d>0$. The network uses $p$-persistent CSMA with $10$ retransmissions in the CRM. Retransmissions improve the performance of a CRM, and are further explained in Remark~\ref{Remark:Retx}. From the simulations shown in Fig.~\ref{Fig:MultiX}, we see that the upper bound of the state magnitude varies with $M$.
\end{example}
The result of the example is not surprising. However, it is not clear how we can identify the network size that can be supported by a given event-triggering policy. Equivalently, for a given network size, how can we identify a stabilizing event-triggering policy? From the example, it is clear that the network interaction between the event-based systems provides the answer. Thus, we now present a model for the network interactions, and derive stability conditions and stabilizing designs using this model.

We use a Markov chain to jointly model the event-triggering policy and CRM, through which each control system interacts with the rest of the network. Since this model applies to each control system, we skip the index $j$ unless we need it to explain the interaction between multiple systems. However, it is useful to keep in mind that every parameter in the following discussion, including probabilities, are unique to each control system, and must be understood to be indexed by $j$. The delay $d$ and an index $S$ are used to denote each state in the Markov chain in Fig.~\ref{Fig:MCAnalysis}. The index $S \in \{I,N,E,T\}$ denotes an idle state ($I$), a non-event state ($N$), an event-state ($E$) and a transmission state ($T$), respectively. We denote the steady state probability of the state $(S,d)$ as $\pi_{_{(S,d)}}$, and compute these values in Lemma~\ref{Lemma:RelAnal}. A successful transmission brings the system to state $(I,0)$, where it awaits the next sampling instant. If the packet is not transmitted, either due to a collision or the lack of an event, the delay increases.

Let us trace through the chain for some delay $d_{k-1}=d-1$, beginning with the plant in the idle state $(I,d-1)$. At the next sampling instant $k$, the state $x_k$ is declared to be an event or a non-event. The control system transitions to $(E,d)$ with event probability $p_{_{\gamma,d}} \triangleq \Pr(\gamma_k=1|d_{k-1}=d-1)$, or to $(N,d)$ with complimentary event probability $q_{_{\gamma,d}} = 1-p_{_{\gamma,d}}$, respectively. From the non-event state $(N,d)$, the system transitions directly to the next idle state $(I,d)$, to wait for the next sampling instant.

An event is sent to the CRM, where it is either transmitted or suppressed. The control system transitions to $(T,d)$ with persistence probability $p_{_{\alpha}}$, or returns to the next idle state $(I,d)$ with complimentary persistence probability $q_{_{\alpha}} = 1 = p_{_{\alpha}}$, respectively. A system in the transmission state, $(T,d)$, sees a busy channel if another control system in the network is in one of its transmission states, $(T,d)$, for any $d > 0$. This happens with probability $p$, and the packet is lost due to a collision. The system then returns to the idle state $(I,d)$. With complimentary probability $q=1-p$, the transmission is successful, and the system transitions to the state $(I,0)$, with the delay reset to zero.

\begin{figure*}[tb]
\begin{center}
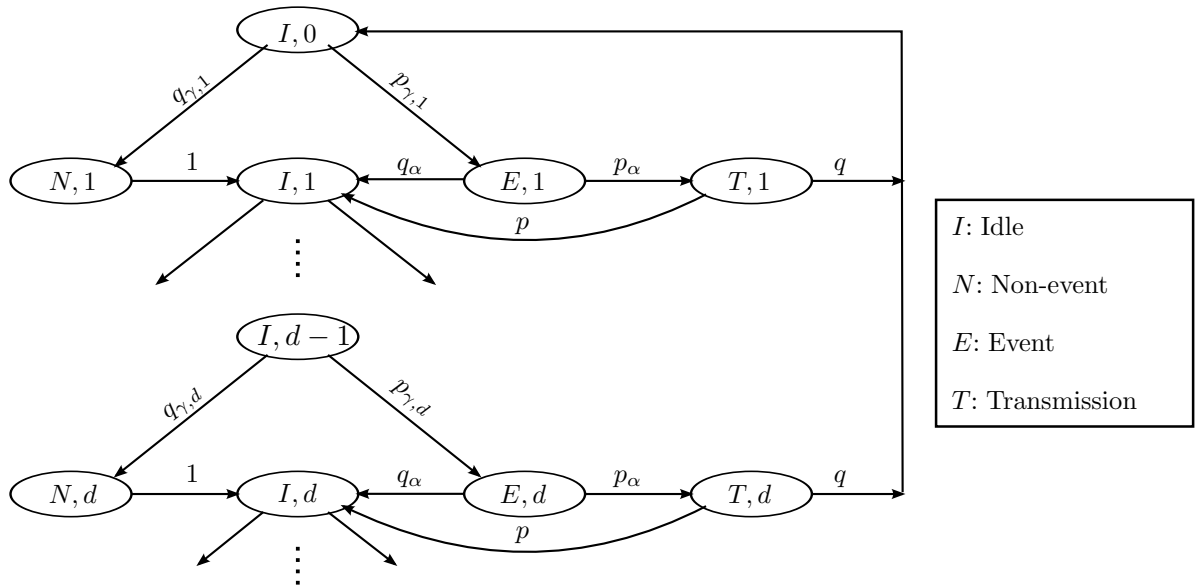
\caption{A Markov chain representation for the event-triggering policy in (\ref{Eq:InnoSched}) and $p$-persistent CSMA with no retransmissions. }
\label{Fig:MCAnalysis}
\end{center}
\end{figure*}

We now present our first assumption, used in the construction of the Markov model.
\begin{assumption}\textbf{\upshape{(Bianchi's Conditional Independence Assumption)}} \label{Assume:Bianchi}
The conditional probability of a busy channel for a node that is ready to transmit (in state $(T,d)$, for $d > 0$), is given by an independent probability $p$. This probability, called the busy channel probability, can be evaluated as
\begin{equation} \label{Eq:pBianchi}
p^{j} = 1 - \prod_{i \neq j, i = 1} ^{M} \big( 1- \sum_{d=1}^{\infty} \pi^{i}_{_{(T,d)}} \big) \; ,
\end{equation}
where, $\pi^{i}_{_{(T,d)}}$ is the steady state probability of the $i^{\textrm{th}}$ control system being in the state $(T,d)$ in the Markov model in Fig.~\ref{Fig:MCAnalysis}.
\end{assumption}

This assumption was introduced by Bianchi in his much-acclaimed analysis of CSMA/CA in $802.11$, and has been verified by many studies~\cite{Bianchi2000,Ramachandran2007,Pollin2008}. For the problem setup considered in this paper, it is verified through simulations in~\cite{Ramesh2011b}.
\begin{remark}\textbf{\upshape{(Extension to CSMA with Multiple Retransmissions)}} \label{Remark:Retx}
The CRM typically makes multiple attempts to transmit the same packet, within a single plant sampling period. This is because the operational time-scale of the CRM is much finer that that of the control system itself. The Markov model presented above can also be used to model such a CRM, by redefining the conditional probability of a busy channel $p$. We do this by first defining a unique conditional probability of a busy channel, $p_r$, for each retransmission attempt $1 \le r \le r_{\max}$. Applying Assumption~\ref{Assume:Bianchi} to each retransmission attempt, the conditional probability of a busy channel in the $r^{\textrm{th}}$ retransmission instant is given by
\begin{equation} \label{Eq:pCondr}
p^{j}_r = 1- \prod_{i \neq j, i = 1}^{M} \big( 1- \sum_{d=1}^{\infty} \pi^{i}_{_{(T,d,r)}}\big) \; ,
\end{equation}
where $\pi^{i}_{_{(T,d,r)}}$ is the steady state probability of the $i^{\textrm{th}}$ control system being in the state $(T,d)$ during the $r^{\textrm{th}}$ retransmission attempt. We now redefine the busy channel probability $p$ to represent an aggregate conditional probability of a busy channel across all the retransmission instants, as given by
\begin{equation} \label{Eq:pBianchiRetx}
p = 1 - \frac{1}{p_{_{\alpha}}} \left( 1 - \prod_{r=1}^{r_{\max}} (1 - p_{_{\alpha}} q_r) \right) \; ,
\end{equation}
where $q_r$ denotes the complimentary probability $1-p_r$.
\end{remark}

We now analyze the reliability of a link, defined as the probability of a successful transmission, in this network.
\begin{lemma}[Reliability Analysis~\cite{Ramesh2011b}] \label{Lemma:RelAnal}
For a network of event-based systems described by \eqref{Eq:StateSpace}--\eqref{Eq:Controller} under Assumption~\ref{Assume:Bianchi}, the network reliability in steady state is given by $\lim_{k \rightarrow \infty} \Pr(\delta_k=1) = \pi_{_{(I,0)}} = q \cdot \sum_{d=0}^{\infty} \pi_{_{(T,d)}}$.
\end{lemma}
\begin{proof}
The steady state distribution of the Markov chain, $\pi_{_{(S,d)}}$ corresponding to the state $(S,d)$, can be calculated when Assumption~\ref{Assume:Bianchi} holds. The steady state probabilities of a node in the states $(I,d)$ and $(T,d)$, respectively, are given by
\begin{align}
\pi_{_{(I,d)}} &= (1-p_{_{\gamma,d}} p_{_{\alpha}} q) \pi_{_{(I,d-1)}} \; , \label{Eq:pId} \\
\pi_{_{(T,d)}} &= p_{_{\gamma,d}} p_{_{\alpha}} \pi_{_{(I,d-1)}} \; . \label{Eq:pRd}
\end{align}
Then, the probability of a successful transmission is given by $\Pr(\delta_k=1) = \pi_{_{(I,0)}}$ and can be obtained by simultaneously solving $\sum_{d=0}^{\infty} \pi_{_{(I,d)}} = 1$ with (\ref{Eq:pBianchi}) or (\ref{Eq:pBianchiRetx}). \hfill \qed
\end{proof}

We now define network steady state as a notion of stability for the network of $M$ control systems.
\begin{definition} \label{Def:NetworkSS}
The network of $M$ control systems is said to be in steady state when $0 \le p < 1$, for the busy channel probability $p$ in (\ref{Eq:pBianchiRetx}).
\end{definition}

When $p=1$, no transmissions occur in the Markov chain in Fig.~\ref{Fig:MCAnalysis}. Thus, the network steady state property simply implies that at least some transmissions occur successfully in the network. We show that network steady state is a necessary condition for Lyapunov mean square stability, for unstable plants.
\begin{proposition} \label{Prop:SteadyStateNecForLMSS}
For unstable plants with spectral radius $\rho(A) > 1$ in the network given by (\ref{Eq:StateSpace})--(\ref{Eq:Controller}), network steady state is a necessary condition for Lyapunov mean square stability, under Assumption~\ref{Assume:Bianchi}.
\end{proposition}
\begin{proof}
The states $(S,d)$, $\forall S \in \{I,N,E,T\}, d \ge 0$, are transient when the busy channel probability $p = 1$, except for the infinite-delay states. For an unstable system, the condition for Lyapunov mean square stability given by (\ref{Eq:LMSS_CE}) cannot be satisfied when $p = 1$, as the variance of the estimation error at infinite delay is not bounded. \hfill \qed
\end{proof}

The above lemma clarifies that a control system cannot be Lyapunov mean square stable without network stability, in the sense defined above. Thus, we begin with the necessary condition that network steady state exists, and then proceed to find conditions for Lyapunov mean square stability. Network steady state is not sufficient to guarantee Lyapunov mean square stability. However, the Lyapunov mean square stability conditions we derive guarantee that network steady state holds.

\begin{remark}\textbf{\upshape{(Necessary Conditions for Bianchi's Assumption to Hold)}}
The existence of the independent process in Bianchi's assumption has been studied in~\cite{Bordenave2010}, among others. The conditions for the decoupling to occur would provide necessary conditions for Proposition~\ref{Prop:SteadyStateNecForLMSS}. An analysis of such conditions for stability of the control system is out of the scope of this paper.
\end{remark}

\begin{remark}\textbf{\upshape{(From Event Thresholds to Event Probabilities)}}
Note that the event thresholds do not directly appear in the Markov model in Fig.~\ref{Fig:MCAnalysis}. The model uses a set of event probabilities $\{p_{_{\gamma,d}}\}$, in place of the event thresholds $\{\Delta_d\}$ to represent the event-triggering policy. The event probabilities are obtained using the event thresholds and the underlying distributions. This alternative representation affords no loss of generality.
\end{remark}

We now summarize the design approach used in this paper. We use a two step strategy to design a stabilizing event-triggering policy. First, we select a stabilizing set of event probabilities, and then, we find a set of event thresholds that result in the designed event probabilities. The motivation for this strategy is because our analysis of a network of event-based systems is parameterized by the event probabilities, rather than the triggering levels, as we saw in the above model. To select suitable event probabilities, we first find conditions that guarantee stability of the control systems in the network, and then outline a design process based on these conditions, as depicted in Fig.~\ref{Fig:FlowChart}.

\begin{figure}[tb]
\begin{center}
\def\svgwidth{8cm}
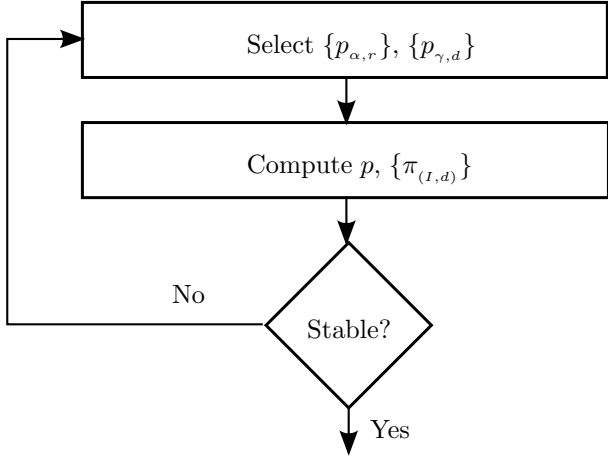
\caption{The event probabilities $\{p_{_{\gamma,d}}\}$ and the persistence probability $\{p_{_{\alpha,r}}\}$ are the inputs to our design process. The stability conditions that we derive in this paper give us a stability guarantee, as an output. The Markov model parameters, $p$ in (\ref{Eq:pBianchiRetx}) and $\{\pi_{_{(I,d)}}\}$ in (\ref{Eq:pId}), are intermediary parameters that must be computed to check the stability conditions. }
\label{Fig:FlowChart}
\end{center}
\end{figure}

\section{Stability Analysis} \label{S:Results}

We use the Markov chain in Fig.~\ref{Fig:MCAnalysis} to analyze the stability of each control system in this section. This analysis includes the effect of all the other control systems in the network through the parameters of the Markov chain. We dispense with the system index $j$ in this section, as our analysis applies to each control system in the network. We begin with the main stability result in this section, and then develop its proof. To arrive at this proof, we examine the underlying density of the estimation error, and construct an auxiliary system that furnishes an upper bound for the variance of each control system.

\subsection{Main Result: Stability Conditions for the Markov Chain}
We begin by presenting one of the main results of the paper. It is a sufficient condition for stabilizing the Markov chain in Fig.~\ref{Fig:MCAnalysis}, in a Lyapunov mean square sense.

\begin{theorem} \label{Thm:SuffGeoSeries}
Consider the network of control systems (\ref{Eq:StateSpace})--(\ref{Eq:Controller}) and suppose Assumption~\ref{Assume:Bianchi} holds. Let $\rho(A_j)$ denote the spectral radius of the $j^{\textrm{th}}$ control system. For $1 \le j \le M$, if
\begin{equation}
\limsup_{d \to \infty} \frac{\pi^{j}_{_{(I,d+1})}}{\pi^{j}_{_{(I,d)}}} < \frac{1}{1+\rho(A_j)^2} \label{Eq:SuffGeoSeries}
\end{equation}
holds, then each of the control systems in the network is Lyapunov mean square stable.
\end{theorem}

The proof is presented in Section~\ref{S:Proof}. The above result requires the probability of the idle states in the tail of the Markov model in Fig.~\ref{Fig:MCAnalysis} to decrease in a stipulated manner, as determined by the spectral radius of each control system $\rho(A)$. The larger the value of $\rho(A)$, the sharper is the mandated fall off in the probabilities of the idle states.

The role of the spectral radius suggests a similarity to other mean square stability results in networked control systems, particularly for packet losses in the sensing or actuation channel~\cite{Gupta2010,Kar2012} and encoder design for data rate limited channels~\cite{Nair2004,Tatikonda2004a,Tatikonda2004b}. The results for packet loss channels specifies a critical probability of loss, beyond which a control system cannot be stabilized in the mean square sense. This result is obtainable only under a Bernoulli packet loss model, which cannot be applied to our problem setup. The stability result in the case of encoder design specifies a stabilizing rate, derived from a source coding perspective. The Markov model we consider in Fig.~\ref{Fig:MCAnalysis} is quite general, and a more specific stability result is difficult to find. In practice, one must find a finite parameterization of the Markov model parameters to obtain a stability condition that can be checked, as illustrated for an event-triggering design that we present in Section~\ref{S:SchedDesign}.

\begin{figure*}[tb]
\centering
\includegraphics*[width=15cm,viewport=0 150 1150 450]{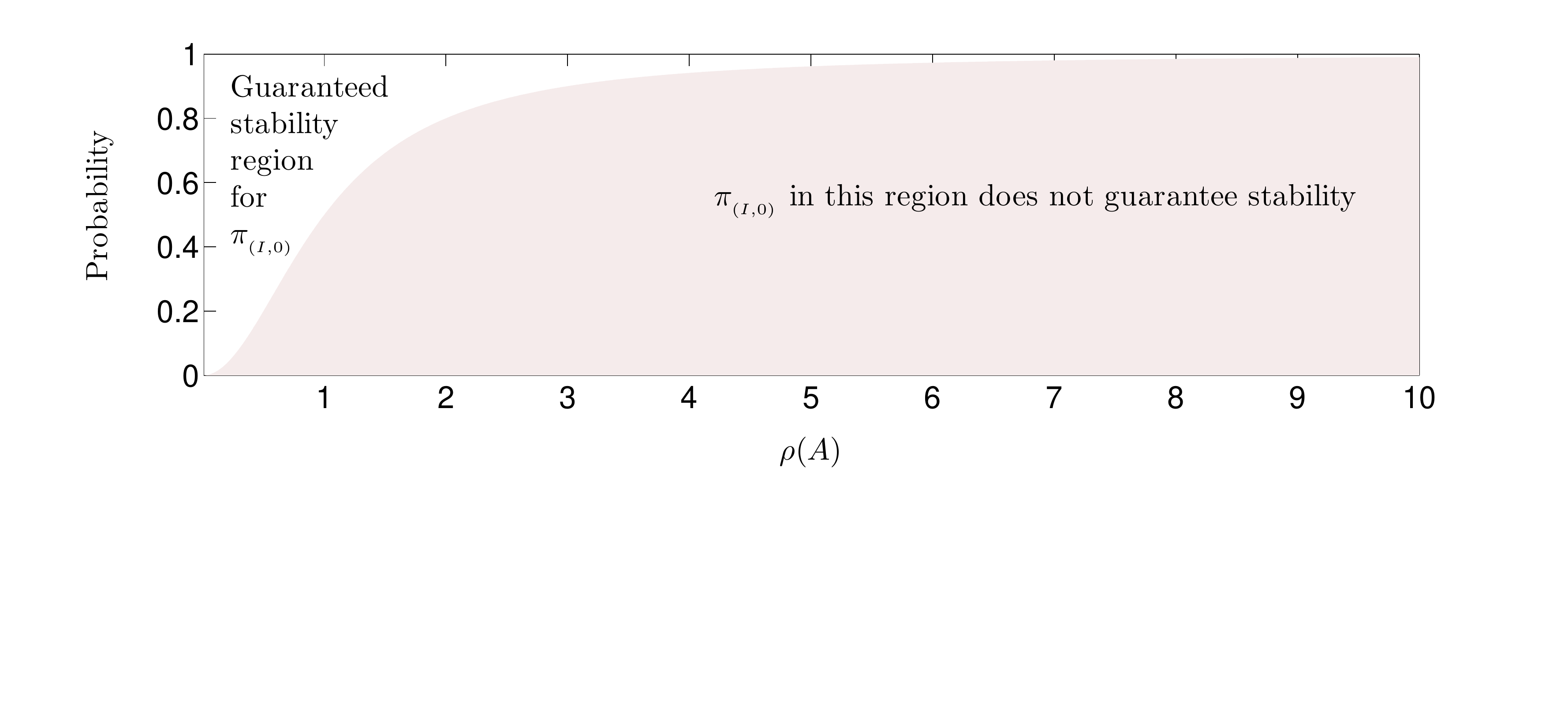}
\caption{A sufficient condition for Lyapunov mean square stability that requires the network reliability, $\pi_{_{(I,0)}}$, to be greater than the line demarcating the regions, with respect to the spectral radius $\rho(A)$. Thus, unstable processes require a higher network reliability to guarantee stability. }
\label{Fig:LowerBoundpI0}
\end{figure*}

\begin{remark}\textbf{\upshape{(Lyapunov mean square stability implies network steady state)}}
Network steady state is not sufficient to guarantee Lyapunov mean square stability. This can be seen by noting that the condition for the busy channel probability, $p < 1$, as required by network steady state, implies that $\pi_{_{(I,d+1)}} < \pi_{_{(I,d)}}$. Thus, network steady state ensures that the loop is sometimes closed, as against the case $p=1$, when the loop is never closed. But, this feedback may not be sufficient to stabilize the control system, in the sense of Definition~\ref{Def:MSS}. However, Lyapunov mean square stability for all $M$ control systems in the network ensures a network steady state, in the sense of Definition~\ref{Def:NetworkSS}. To see this, note that $\pi_{_{(I,d+1)}} < \frac{\pi_{_{(I,d)}}}{1+\rho(A)^2} < \pi_{_{(I,d)}}$, for all $\rho(A)>0$. Hence, network steady state is indeed achieved by the control systems in stability.
\end{remark}

\begin{remark} \textbf{\upshape{(Guaranteeing Stability for Unstable Processes)}} \label{Remark:pI0}
Let us assume that we choose the event probabilities such that (\ref{Eq:SuffGeoSeries}) is true for all $d \ge 0$, as opposed to the tail of the sequence alone. Then, using (\ref{Eq:SuffGeoSeries}) in $\sum_{d=0}^{\infty} \pi_{_{(I,d)}} = 1$, we get a lower bound for the network reliability as $\pi_{_{(I,0)}} > \rho(A)^2/(1+\rho(A)^2)$. This implies that the network reliability must lie above the line shown in Fig.~\ref{Fig:LowerBoundpI0}, and thus, unstable processes require a higher network reliability to guarantee stability.
\end{remark}

\begin{remark}[Role of the Persistence Probability]
Using the recursive relationship for the idle state probabilities $p_{_{(I,d+1)}}$ in (\ref{Eq:pId}), along with the sufficient condition in (\ref{Eq:SuffGeoSeries}), we obtain
\begin{equation*}
\limsup_{d \to \infty} \frac{(1-p_{_{\gamma,d+1}} p_{_{\alpha}} q) \pi_{_{(I,d)}}}{\pi_{_{(I,d)}}} < \frac{1}{1+\rho(A)^2} \; ,
\end{equation*}
which can be rearranged to obtain $\limsup_{d \to \infty} p_{_{\gamma,d}} q > \kappa_{_{\alpha}}$, where $\kappa_{_{\alpha}} = \frac{1}{p_{_{\alpha}}} \frac{\rho(A)^2}{1+\rho(A)^2}$. The value of $\kappa_{_{\alpha}}$ can be tuned by varying the persistence probability $p_{_{\alpha}}$. A small value for $p_{_{\alpha}}$ can increase the lower bound for $p_{_{\gamma,d}} q$, which in turn can improve the network reliability.
\end{remark}

We first attempt to prove Theorem~\ref{Thm:SuffGeoSeries} directly by examining the underlying density of the estimation error. This is a difficult approach, as we show in Section~\ref{SS:PDFevolve}. Then, we construct auxiliary systems in Sections~\ref{SS:PDFmaj} and~\ref{SS:PDFmajVector}, and use these systems to prove Theorem~\ref{Thm:SuffGeoSeries} in Section~\ref{S:Proof}.

\subsection{A Difficult Direct Approach} \label{SS:PDFevolve}
We seek an expression for the variance of the estimation error. Let us associate with each state $(S,d)$ for $S \in \{I,N,E,T\}$ a probability density function (PDF) for the estimation error (filtered or predicted) at the controller, denoted by $\phi_{_{(S,d)}}$, for the appropriate estimation error corresponding to the state $(S,d)$ of the Markov model.

Then, the variance of the estimation error conditioned on a delay $d$ is given by $\tr\{P_d\}$, where $P_d = \int_{-\infty}^{\infty} \tilde{x} \tilde{x}^T \phi_{_{(I,d)}}(\tilde{x}) d \tilde{x}$. Marginalizing over the idle state distribution, we get
\begin{equation} \label{Eq:EstErrCov}
\tr\{\E[P_{^{k|k}}]\} = \sum_{d=0}^{\infty} \tr\{P_d\} \pi_{_{(I,d)}} \; .
\end{equation}
The above expression is simple, but the PDFs can be hard to evaluate. To see why, let us look at the evolution of these PDFs as the delay $d$ increases. For $d=0$, $\phi_{_{(I,0)}} := \lim_{k \to \infty} \phi(\tilde{x}^c_{^{k|k-1}} | \delta_{k-1} = 1)$ is the PDF associated with the predicted estimate, one step after a transmission. Clearly, $\phi_{_{(I,0)}} = \phi_N(\Sigma_w)$, where $\phi_N$ is the PDF of a normal distribution with covariance $\Sigma_w$. For any delay $d$, the PDFs associated with the event ($\gamma_k = 1$) and non-event ($\gamma_k = 0$) states are truncated versions of the PDF associated with the previous idle state. They can be defined as $\phi_{_{(N,d)}} := \lim_{k \to \infty} \phi(\tilde{x}^c_{^{k|\tau_{k-1}}} | \gamma_k = 0, d_{k-1} = d-1)$ and $\phi_{_{(E,d)}} := \lim_{k \to \infty} \phi(\tilde{x}^c_{^{k|\tau_{k-1}}} | \gamma_k = 0, d_{k-1} = d-1)$, respectively. Thus, we get
\begin{align}
\phi_{_{(N,d)}} &= \begin{cases}
\frac{\phi_{_{(I,d-1)}}(\tilde{x})}{q_{_{\gamma,d}}} & |\tilde{x}| \le \Delta_d \; ,\\
0 & \textrm{otherwise} \; ,
\end{cases} \; , \label{Eq:pdfNd}
\end{align}
\begin{align}
\phi_{_{(E,d)}} &= \begin{cases}
\frac{\phi_{_{(I,d-1)}}(\tilde{x})}{p_{_{\gamma,d}}} & |\tilde{x}| > \Delta_d \; , \\
0 & \textrm{otherwise} \; ,
\end{cases} \label{Eq:pdfEd}
\end{align}
where, $q_{_{\gamma,d}} = \int_{-\Delta_d}^{\Delta_d} \phi_{_{(I,d-1)}}(\tilde{x}) d \tilde{x}$ is the probability of a non-event and $p_{_{\gamma,d}} = 1-q_{_{\gamma,d}}$ is the probability of an event.

Then, let us denote $e_d$ as the innovations process that does not get transmitted after a delay $d$, and denote its PDF as $\phi^e_{_{(I,d)}}:= \lim_{k \to \infty} \phi(\tilde{x}^c_{^{k|k}} | \delta_k=0, d_k=d)$. This PDF can be rewritten as
\begin{align}
\phi^e_{_{(I,d)}} 
= &\phi_{_{(N,d)}}(\tilde{x}) \cdot \frac{q_{_{\gamma,d}}}{q_{_{\gamma,d}}+p_{_{\gamma,d}}(q_{_{\alpha}} + p p_{_{\alpha}})} \\
&\quad + \phi_{_{(E,d)}}(\tilde{x}) \cdot \frac{(q_{_{\alpha}} + p p_{_{\alpha}}) p_{_{\gamma,d}}}{q_{_{\gamma,d}}+p_{_{\gamma,d}}(q_{_{\alpha}} + p p_{_{\alpha}})} \notag
\end{align}
Substituting for $\phi_{_{(N,d)}}$ and $\phi_{_{(E,d)}}$ from (\ref{Eq:pdfNd}) and (\ref{Eq:pdfEd}), respectively, we obtain
\begin{align}
\phi^e_{_{(I,d)}} = &\begin{cases}
\phi_{_{(I,d-1)}}(\tilde{x}) \cdot \frac{1}{q_{_{\gamma,d}}+p_{_{\gamma,d}}(q_{_{\alpha}} + p p_{_{\alpha}})} & |\tilde{x}| \le \Delta_d \\
\phi_{_{(I,d-1)}}(\tilde{x}) \cdot \frac{(q_{_{\alpha}} + p p_{_{\alpha}})}{q_{_{\gamma,d}}+p_{_{\gamma,d}}(q_{_{\alpha}} + p p_{_{\alpha}})} & |\tilde{x}| > \Delta_d
\end{cases} \label{Eq:phied}
\end{align}
Finally, the PDF of the idle state with delay $d$ is denoted $\phi_{_{(I,d)}} := \lim_{k \to \infty} \phi(\tilde{x}^c_{^{k+1|\tau_k}}| d_k = d)$. For a plant with an invertible $A$ matrix\footnote{If the matrix $A$ is non-invertible, the PDF corresponding to the idle state is no longer defined in $\mathbb{R}^n$. A measure on the subspace orthogonal to the null-space of $A$ is absolutely continuous w.r.t. the Lebesgue measure, and the PDF, along with the expected value or covariance, is defined using this measure. Thus, the approach presented in this paper is applicable for plants with non-invertible $A$ matrices. However, for ease of exposition, we present the results assuming that $A$ is invertible.}, we can use the state update equation in (\ref{Eq:StateSpace}) to find an expression for $\phi_{_{(I,d)}}$ as
\begin{equation} \label{Eq:UpdatePDF}
\phi_{_{(I,d)}} = \frac{1}{|\det(A)|} \phi^e_{_{(I,d)}}(A^{-1} \tilde{x}) \ast \phi_N(\Sigma_w) \; ,
\end{equation}
where $\ast$ denotes the convolution operator. 

The above operations must be performed recursively, to obtain the PDF associated with the state $(I,d)$. This computation is in general hard. Hence, we find auxiliary systems that result in upper bounds for the variance of the estimation error in the following subsections.

\subsection{Auxiliary PDFs for First-Order Systems} \label{SS:PDFmaj}

We wish to find an upper bound for the variance of the estimation error in the idle states of the Markov chain. To do this, we must first find a sequence of PDFs, $\hat{\phi}_{_{(I,d)}}$, that are more `spread out' than the PDFs $\phi_{_{(I,d)}}$. We use stochastic majorization to do this. Our approach, in this section, is restricted to first-order systems, due to a symmetry requirement on PDFs. In the following section, we extend these results to higher-order systems using other methods that give us more conservative results. We need the following notation and definitions, adapted from~\cite{Hajek2008}, to define majorization.

\begin{definition} \textbf{\upshape{(Symmetric Non-increasing Function)}} \label{Def:SymmNIfunc}
A function $f:\mathbb{R}^n \rightarrow \mathbb{R}$ is said to be symmetric non-increasing if $f(x) = \phi(|x|)$, for some non-increasing function $\phi$ on $\mathbb{R}^+$, where $|x|$ denotes the Euclidean norm of $x \in \mathbb{R}^n$.
\end{definition}

\begin{figure*}[tb]
\begin{center}
\def\svgwidth{12cm}
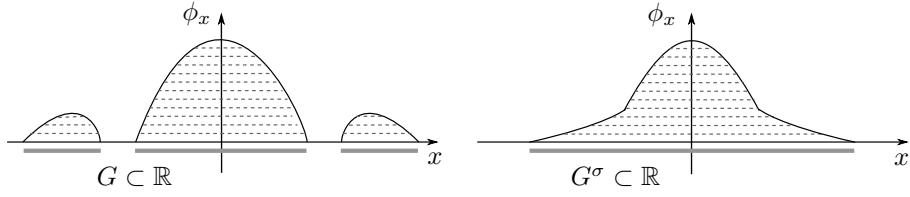
\caption{An illustration of symmetric rearrangement in Definition~\ref{Def:SymmRearrangement}. The level sets of the PDF on the left are symmetrically placed around the origin, to obtain the symmetric rearrangement on the right. An example of a level set is $G$, given by the union of the three shaded segments on the left. The symmetric rearrangement of this level set results in $G^\sigma$ on the right, with length equal to the sum of the lengths of the three shaded segments on the left. From this figure, it is easy to see how the variance of the symmetrically rearranged PDF is always less than the variance of the original PDF (Lemma~\ref{Lemma:EstErrCovOrdering}).}
\label{Fig:SymmRearrange}
\end{center}
\end{figure*}
Given any integrable, non-negative function, we wish to `rearrange' the function to obtain a symmetric non-increasing function. The exact sense in which we rearrange the function is defined below. We begin with a definition for the symmetric rearrangement of a Borel set. Then, we apply this definition to the level sets of a non-negative function, and obtain its symmetric rearrangement. We illustrate this notion in Fig.~\ref{Fig:SymmRearrange}.
\begin{definition} [Symmetric Rearrangement] \label{Def:SymmRearrangement}
Let $G \in \mathcal{B}$ be a Borel set in $\mathbb{R}^n$, with finite Lebesgue measure $\mathcal{L}(G)$. The \emph{symmetric rearrangement} of $G$, denoted by $G^\sigma$, is the open ball in $\mathbb{R}^n$ centered at the origin, with measure $\mathcal{L}(G^\sigma) = \mathcal{L}(G)$.

For an integrable, non-negative function $h$ on $\mathbb{R}^n$, its symmetric non-increasing rearrangement, denoted $h^\sigma$ is given by
\begin{equation} \label{Eq:SymmNDRearrangement}
h^\sigma(x) \triangleq \int_{0}^{\infty} I_{{\{ x': h(x') > l \}}^{\sigma}}(x) dl \; ,
\end{equation}
where $I_{\{x': x' \in G\}^\sigma}(x) = \{x: x \in G^\sigma \}$, denotes the set of elements belonging to the symmetric rearrangement of its argument set $G$.
\end{definition}

We now define majorization with the help of the distribution functions corresponding to the symmetrically rearranged densities.
\begin{definition} [Majorization] \label{Def:Maj}
Given two PDFs $\phi_a$ and $\phi_b$ on $\mathbb{R}^n$, we say that $\phi_a$ majorizes $\phi_b$, denoted as $\phi_a \succ \phi_b$, if
\begin{equation*}
\int_{|x| \le \rho} \phi_a^\sigma(x) d x \ge \int_{|x| \le \rho} \phi_b^\sigma(x) d x \; , \; \forall \rho \ge 0 \; .
\end{equation*}
\end{definition}
Thus, $\phi_a$, as per the above definition is more contained, or less spread out, than $\phi_b$. Some results involving the majorization operator are listed in Appendix~\ref{App:MajLemmas}. The most important consequence for us is that we obtain an upper bound for the estimation error variance. This is stated below.
\begin{lemma} \textbf{\upshape{(Ordering of Estimation Error Variance)}} \label{Lemma:EstErrCovOrdering}
If $\phi_a$ and $\phi_b$ are symmetric non-increasing PDFs on $\mathbb{R}^n$ such that $\phi_a \succ \phi_b$, then $\int_{-\infty}^{\infty} |x|^2 \phi_a(x) dx \le \int_{-\infty}^{\infty} |x|^2 \phi_b(x) dx$.
\end{lemma}
\begin{proof}
Use $h(x) = |x|^2$ in Lemma~\ref{Lemma:Maj4} to obtain the results. \hfill \qed
\end{proof}

We now describe the PDFs that we are interested in, as adapted from~\cite{Lipsa2011}.
\begin{definition} [Neat PDF] \label{Def:NeatPDF}
We say that a PDF $\phi$ is \emph{neat} if it is quasi-concave and if there exists a real number $r$ such that $\phi$ is non-decreasing on $(-\infty,r]$ and non-increasing on $[r,\infty)$.
\end{definition}

Note that PDFs on $\mathbb{R}$ are symmetric non-increasing if and only if they are neat and even. Thus, for neat PDFs, the definition of majorization can be directly applied to the PDF itself. 
Using Definition~\ref{Def:Maj}, we find a more spread out $\hat{\phi}_{_{(I,d)}}$, as stated below.
\begin{lemma} \label{Lemma:HatPhi}
Let the auxiliary PDF, $\hat{\phi}_{_{(I,d)}}$, be defined by the recursive relation
\begin{equation*}
\hat{\phi}_{_{(I,d)}} = \frac{1}{A}\hat{\phi}_{_{(I,d-1)}} \ast \phi_N \; ,
\end{equation*}
with $\hat{\phi}_{_{(I,0)}} = \phi_N$. Then, $\phi_{_{(I,d)}} \succ \hat{\phi}_{_{(I,d)}}$ for all $d \ge 0$.
\end{lemma}
\begin{proof}
We show this using induction. Trivially, at $d=0$, $\phi_{_{(I,0)}} = \hat{\phi}_{_{(I,0)}} = \phi_N$. Let us assume that, for some $d$, $\phi_{_{(I,d)}} \succ \hat{\phi}_{_{(I,d)}}$. Then, from (\ref{Eq:phied}), we can show that $\phi^e_{_{(I,d+1)}} \succ \phi_{_{(I,d)}}$. To see this, recall that $\phi^e_{_{(I,d+1)}}$ is obtained by appropriately combining the truncated PDFs for the event and non-event states, as shown in (\ref{Eq:phied}). Then, we have
\begin{itemize}
\item For $|e| \le \Delta_{d+1}$, we have
\begin{equation*}
\int_{-e}^{e} \frac{\phi_{_{(I,d)}}(\tilde{x})}{q_{_{\gamma,d+1}}+p_{_{\gamma,d+1}}(q_{_{\alpha}} + p p_{_{\alpha}})} d \tilde{x} \ge \int_{-e}^{e} \phi_{_{(I,d)}}(\tilde{x}) d \tilde{x} \; ,
\end{equation*}
because $q_{_{\gamma,d+1}}+p_{_{\gamma,d+1}}(q_{_{\alpha}} + p p_{_{\alpha}}) \le 1$.
\item For $|e| > \Delta_{d+1}$, we have
\begin{align*}
\int_{-\infty}^{-e} &\phi_{_{(I,d)}} (\tilde{x}) \frac{(q_{_{\alpha}} + p p_{_{\alpha}})}{q_{_{\gamma,d+1}} + p_{_{\gamma,d+1}}(q_{_{\alpha}} + p p_{_{\alpha}})} d \tilde{x} \\
&+ \int_{e}^{\infty} \phi_{_{(I,d)}}(\tilde{x}) \frac{(q_{_{\alpha}} + p p_{_{\alpha}})}{q_{_{\gamma,d+1}} + p_{_{\gamma,d+1}}(q_{_{\alpha}} + p p_{_{\alpha}})} d \tilde{x} \\
&\le \int_{-\infty}^{-e} \phi_{_{(I,d)}}(\tilde{x}) d \tilde{x} + \int_{e}^{\infty} \phi_{_{(I,d)}}(\tilde{x}) d \tilde{x} \; ,
\end{align*}
because $\frac{(q_{_{\alpha}} + p p_{_{\alpha}})}{q_{_{\gamma,d+1}} + p_{_{\gamma,d+1}}(q_{_{\alpha}} + p p_{_{\alpha}})} \le 1$.
\end{itemize}

Since $\phi^e_{_{(I,d+1)}} \succ \phi_{_{(I,d)}}$ and $\phi_{_{(I,d)}} \succ \hat{\phi}_{_{(I,d)}}$, we have
\begin{align*}
\phi^e_{_{(I,d+1)}} &\succ \hat{\phi}_{_{(I,d)}} \\
\frac{1}{A} \phi^e_{_{(I,d+1)}}(\frac{\tilde{x}}{A}) &\succ \frac{1}{A} \hat{\phi}_{_{(I,d)}}(\frac{\tilde{x}}{A}) \\
\phi_N \ast \frac{1}{A} \phi^e_{_{(I,d+1)}}(\frac{\tilde{x}}{A}) &\succ \phi_N \ast \frac{1}{A} \hat{\phi}_{_{(I,d)}}(\frac{\tilde{x}}{A}) \; ,
\end{align*}
where, the last two expressions are obtained from the results of Lemma~\ref{Lemma:Maj3} and Lemma~\ref{Lemma:SymmNIConv}, respectively. Hence, $\phi_{_{(I,d+1)}} \succ \hat{\phi}_{_{(I,d+1)}}$. \hfill \qed
\end{proof}

\begin{remark}\textbf{\upshape{(Worst-Case Evolution of the System)}}
The PDFs~given by $\hat{\phi}_{_{(I,d)}}$ correspond to the evolution of the control system when the busy channel probability $p=1$, in the Markov model. For such systems, no event is successfully transmitted, and hence, the density of the tail ($|\tilde{x}| > \Delta$) is never reduced, due to a perpetually busy channel ($p=1$). Thus, the gaussian property of the estimation error is retained and its PDF is given by $\hat{\phi}_{_{(I,d)}}$.
\end{remark}

\subsection{Auxiliary PDFs for Higher-Order Systems} \label{SS:PDFmajVector}
In this section, we find an upper bound for the variance of the estimation error, for higher-order systems. The PDF of the state for such control systems need not be symmetric, and hence, the results developed in the previous section cannot be directly applied to such systems. We denote the multivariate PDFs in this section with $\Phi$ in place of $\phi$.

We now find an upper bound for the variance of the estimation error associated with the Markov chain states $(I,d)$ for all $d \ge 0$, by finding suitable PDFs $\hat{\Phi}_{_{(I,d)}}$. We first define the matrix $\bar{A} = \rho(A) I_n$, where $\rho(A)$ is the spectral radius of $A$ and $I_n \in \mathbb{R}^{n \times n}$ is an identity matrix. Let $\mathrm{var}(\Phi)$ denote the variance of the PDF $\Phi$. We now have the following bound on the variance, cf.~Lemma~\ref{Lemma:HatPhi}. 
\begin{lemma} \label{Lemma:HatPhiVec}
Let the auxiliary PDF, $\hat{\Phi}_{_{(I,d)}}$, be defined by the recursive relation
\begin{equation*}
\hat{\Phi}_{_{(I,d)}}(\tilde{x}) = \frac{1}{|\det(\bar{A})|}\hat{\Phi}_{_{(I,d-1)}}(\bar{A}^{-1} \tilde{x}) \ast \Phi_N \; ,
\end{equation*}
with $\hat{\Phi}_{_{(I,0)}} = \phi_N(\Sigma_w)$. Then, $\mathrm{var}(\phi_{_{(I,d)}}) \le \mathrm{var}(\hat{\Phi}_{_{(I,d)}})$ for all $d \ge 0$.
\end{lemma}
\begin{proof}
We show this using induction. Trivially, at $d=0$, $\Phi_{_{(I,0)}} = \hat{\Phi}_{_{(I,0)}} = \Phi_N(\Sigma_w)$. Thus, the variances are equal for $d=0$. Let us assume that, for some $d > 0$, $\mathrm{var}(\phi_{_{(I,d)}}) \le \mathrm{var}(\hat{\Phi}_{_{(I,d)}})$. We denote the variance of $e_d$, the innovations process that does not get transmitted after a delay $d$, as $\mathrm{var}(\phi^{e}_{_{(I,d+1)}})$, following the notation in (\ref{Eq:phied}). Then, we show that $\mathrm{var}(\phi^{e}_{_{(I,d+1)}}) \le \mathrm{var}(\Phi_{_{(I,d)}})$ in Lemma~\ref{Lemma:ReducedVariance}. 
Combining this with our induction assumption, we obtain $\mathrm{var}(\hat{\Phi}_{_{(I,d)}}) \ge \mathrm{var}(\phi_{_{(I,d)}}) \ge \mathrm{var}(\phi^{e}_{_{(I,d+1)}})$.

At the next sampling instant, the state is updated according to the state-space model, with a linear transformation and an addition of process noise. The linear transformation of a random vector results in the PDFs denoted $\Phi^{e,+}_{_{(I,d+1)}}$ for the original system, and $\hat{\Phi}^{+}_{_{(I,d)}}$ for the auxiliary system. The transformed PDFs are given by
\begin{align*}
\Phi^{e,+}_{_{(I,d+1)}} &= \frac{1}{|\det(A)|} \Phi^e_{_{(I,d+1)}}(A^{-1} \tilde{x}) \\
\hat{\Phi}^{+}_{_{(I,d)}} &= \frac{1}{|\det(\bar{A})|} \hat{\Phi}_{_{(I,d)}}(\bar{A}^{-1} \tilde{x}) \; .
\end{align*}
The variances can be written as
\begin{align*}
\mathrm{var}(\Phi^{e,+}_{_{(I,d+1)}}) &= \tr\{ A \Sigma^e_{_{(I,d+1)}} A^T\} = \tr\{ \Sigma^e_{_{(I,d+1)}} A^T A \} \\
\mathrm{var}(\hat{\Phi}^{+}_{_{(I,d)}}) &= \tr\{ \bar{A} \hat{\Sigma}_{_{(I,d)}} \bar{A}^T\} = \tr\{ \hat{\Sigma}_{_{(I,d)}} \bar{A}^T \bar{A} \} \; ,
\end{align*}
where $\Sigma^e_{_{(I,d+1)}}$ and $\hat{\Sigma}_{_{(I,d)}}$ are the covariance matrices associated with PDFs $\Phi^e_{_{(I,d+1)}}$ and $\hat{\Phi}_{_{(I,d)}}$, respectively. Now, note that $\Sigma^e_{_{(I,d+1)}}$ and $A^T A$ are symmetric matrices, and that their product $\Sigma^{e,+}_{_{(I,d+1)}}$ is also a symmetric matrix. Thus, the matrices commute, and we can apply the spectral value inequality to obtain $\rho(\Sigma^{e,+}_{_{(I,d+1)}}) \le \rho(\Sigma^e_{_{(I,d+1)}}) \cdot \rho(A^T A)$. Furthermore, $\tr\{ \hat{\Sigma}_{_{(I,d)}} \bar{A}^T \bar{A} \} = \rho^2(A) \mathrm{var}(\hat{\Phi}_{_{(I,d)}})$. Combining these facts, we obtain $\mathrm{var}(\phi^{e,+}_{_{(I,d+1)}}) \le \mathrm{var}(\hat{\Phi}^{+}_{_{(I,d)}})$.

We have not yet accounted for the addition of process noise in the state update. This operation results in an addition of a constant term $\tr\{R_w\}$, corresponding to the variance of the process noise $w_k$, to both the original and auxiliary system. Thus, the variance ordering is preserved, and we have the desired result $\mathrm{var}(\phi_{_{(I,d+1)}}) \le \mathrm{var}(\hat{\Phi}_{_{(I,d+1)}})$. \hfill \qed
\end{proof}

\begin{remark}[Lossy Network as Upper Bound]
The PDFs of the auxiliary systems are used along with the probabilities in the Markov chain in Fig.~\ref{Fig:MCAnalysis} to upper bound the variance of the estimation error of the control system. The resulting approximation describes the evolution of a system with a lossy sensor link, albeit with a loss probability that varies with delay, as shown in Fig.~\ref{Fig:MClossyNW}. The loss probability is given by $p_{_{l,d}} = 1-p_{_{\gamma,d}} p_{_{\alpha}} q$. The estimation error covariance of this system for zero delay is clearly $\hat{P}_0 = \Sigma_w$, and for all other delays $d > 0$, is given by
\begin{equation} \label{Eq:EstErrCov_LossyNW}
\hat{P}_{d} = \rho(A)^2 \hat{P}_{d-1} + \Sigma_w \; .
\end{equation}
\end{remark}

\begin{figure*}[tb]
\begin{center}
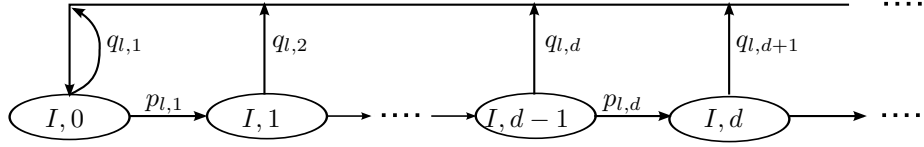
\caption{The majorized PDFs, from Lemma~\ref{Lemma:HatPhi}, along with the probabilities from the original Markov chain are combined to describe a lossy sensor link. The estimation error variance of this system is an upper bound for the control system. }
\label{Fig:MClossyNW}
\end{center}
\end{figure*}

\subsection{Proof of Theorem~\ref{Thm:SuffGeoSeries}} \label{S:Proof}

Let us now prove the stability conditions in Theorem~\ref{Thm:SuffGeoSeries} using the auxiliary systems we have identified.
\begin{proof}
The estimation error covariance can be bounded from above, using the approximations from Lemma~\ref{Lemma:HatPhi}, as
\begin{equation*}
\tr\{\E[P_{^{k|k}}]\} = \sum_{d=1}^{\infty} \pi_{_{(I,d)}} \tr\{P_d\}
\le \sum_{d=1}^{\infty} \pi_{_{(I,d)}} \tr\{\hat{P}_d\} \; . 
\end{equation*}
For this expression to be bounded~\cite{Rudin1976}, we require
\begin{equation*}
\limsup_{d \to \infty} \frac{\pi_{_{(I,d+1)}}\tr\{\hat{P}_{d+1}\}}{\pi_{_{(I,d)}}\tr\{\hat{P}_{d}\}} < 1 \; .
\end{equation*}
Since $\tr\{\hat{P}_d\} = \Sigma_w (1+\rho(A)^2+\dots+\rho(A)^{2(d-1)})$, the left hand side of the above inequality can be written as
\begin{align*}
&\limsup_{d \to \infty}  \frac{\pi_{_{(I,d+1)}}}{\pi_{_{(I,d)}}} \bigg[ 1 + \rho(A)^2 \frac{\rho(A)^{2(d-1)}}{1+\rho(A)^2+\dots+\rho(A)^{2(d-1)}} \bigg] \\
\le &\limsup_{d \to \infty} \frac{\pi_{_{(I,d+1)}}}{\pi_{_{(I,d)}}} [ 1 + \rho(A)^2 ] \; .
\end{align*}
By requiring the last expression to be strictly less than $1$, we satisfy the condition  in (\ref{Eq:SuffGeoSeries}) required to obtain Lyapunov mean square stability. \hfill \qed
\end{proof}

\section{Event-Triggering Policy Synthesis} \label{S:SchedDesign}
We now look at the problem of designing stabilizing event-triggering policies. In particular, how should the event probabilities $\{p_{_{\gamma,d}}\}$ be chosen as a function of $d$ to achieve Lyapunov mean square stability? We can immediately think of three possible ways to let the event probabilities vary with the delay: holding it a constant, additively increasing or decreasing it, or multiplicatively increasing or decreasing it. We discuss the constant-probability policy in detail and identify stability conditions for such policies. We then discuss the feasibility of the other policies briefly.

\subsection{Constant-Probability Policy}
The constant-probability policy provides a constant event probability for all delays, i.e. $p_{_{\gamma,d}} = p_{_{\gamma}}$, for all $d>0$. Using the lossy network model from Section~\ref{S:Results}, we identify stability conditions for this particular policy.

\begin{theorem} \label{Thm:ConstLaw}
For the control system given by (\ref{Eq:StateSpace})--(\ref{Eq:Controller}), a sufficient condition for Lyapunov mean square stability for the constant-probability event-triggering policy is given by
\begin{equation} \label{Eq:StabilityConditions_ConstLaw}
p_{_{\gamma}} \bigg( \sum_{r=1}^{r_{\max}} \big ( q_{_{r}} \cdot \prod_{r=1}^{r_{\max}-1} (1-p_{_{\alpha}} q_{_{r}})\big) \bigg) > \frac{1}{p_{_{\alpha}}} (1-\frac{1}{\rho(A)^2}) \; .
\end{equation}
\end{theorem}
\begin{proof}
Using a constant scheduling law, note that the lossy network model has a constant loss probability, $p_{_{l}} = 1-p_{_{\gamma}} p_{_{\alpha}} q$. Thus, the estimation error variance in this model, given in (\ref{Eq:EstErrCov_LossyNW}), converges if
\begin{equation*}
p_{_{l}} \rho(A)^2 < 1 \; .
\end{equation*}
Substituting for $p_{_{l}}$ above from (\ref{Eq:pBianchiRetx}), and rearranging, we obtain the condition in (\ref{Eq:StabilityConditions_ConstLaw}). \hfill \qed
\end{proof}

The constant-probability policy results in simple, closed-form expressions for the probability of successful transmission $p_{_{I,0}}$ and the loss probability $p_{_{l}}$. To see this, note that the sum of the probabilities of the idle states is given by the sum of a geometric series, $\sum_{d=0}^\infty \pi_{_{(I,d)}} = \pi_{_{(I,0)}} \sum_{d=0}^\infty p_{_{l}}^d$. Thus, we have
\begin{equation} \label{Eq:pI0_ConstLaw}
\pi_{_{(I,0)}} = p_{_{\gamma}} p_{_{\alpha}} q \; ,
\end{equation}
using the expression for the loss probability, where $q$ is the complimentary busy channel probability $1-p$. The conditional probability of a busy channel in each retransmission attempt can be computed using (\ref{Eq:pCondr}) and (\ref{Eq:pRd}) as
\begin{align*}
\sum_{d=1}^\infty p_{_{T,d,r}} &= p_{_{\gamma}} p_{_{\alpha}} \prod_{s=1}^r (1 - p_{_{\alpha}} q_s) \sum_{d=1}^\infty \pi_{_{(I,d-1)}}
\end{align*}
which leads to
\begin{align}
p_r &= 1 - \bigg( 1 - p_{_{\gamma}} p_{_{\alpha}} \prod_{s=1}^r (1 - p_{_{\alpha}} q_s) \bigg)^{M-1} \; , \; r \in \{1,\dots,r_{\max}\} \label{Eq:pr_ConstLaw}
\end{align}

Equations (\ref{Eq:pI0_ConstLaw})--(\ref{Eq:pr_ConstLaw}) along with (\ref{Eq:StabilityConditions_ConstLaw}), gives us the event probability $p_{_{\gamma,d}}$ for a given persistence probability $p_{_{\alpha}}$, that guarantees Lyapunov mean square stability. The flowchart in Fig.~\ref{Fig:FlowChart} gives us a set of $p_{_\alpha}$ and $p_{_{\gamma}}$ that generate a stable system. We present an example of this design procedure in Section~\ref{S:Results}.


\subsection{Additive-Probability Policy}
The additive-probability policy is designed to provide an additive increase/decrease in the event probability with delay, i.e., $p_{_{\gamma,d}} = p_{_{\gamma,d-1}} + \nu$, for $\nu \gtrless 0$. Note that $\lim_{d \to \infty} p_{_{\gamma,d}} \rightarrow \infty$ for $\nu > 0$ and $\lim_{d \to \infty} p_{_{\gamma,d}} \rightarrow -\infty$ for $\nu < 0$. Thus, we let the additive terms decrease in magnitude, such that $\sum_{d=1}^{\infty} \nu_d$ is bounded and $\lim_{d \to \infty} p_{_{\gamma,d}} < 1$. 
Many such examples can be found. A simple example is
\begin{equation} \label{Eq:IncSchedLaw}
p_{_{\gamma,d}} = p_{_{\gamma,1}} + \eta + \eta^2 + \dots + \eta^{d-1} \; ,
\end{equation}
which gives rise to an increasing law when $\eta > 0$, and a decreasing law when $\eta < 0$. Thus, $p_{_{\gamma,\infty}} = p_{_{\gamma,1}} + \frac{\eta}{1-\eta}$, and $p_{_{\gamma,1}}$ and $\eta$ must be chosen such that $p_{_{\gamma,1}} < 1$ and $p_{_{\gamma,\infty}} < 1$. Then, we can apply Theorem~\ref{Thm:SuffGeoSeries} to identify designs that are guaranteed to result in Lyapunov mean square stability.

\subsection{Exponential-Probability Policy}
The exponential-probability policy is designed to provide an exponential increase/decrease in the event probability with delay, i.e., $p_{_{\gamma,d}} = \mu p_{_{\gamma,d-1}}$, for $\mu < 1$. Note that if $\mu > 1$, $p_{_{\gamma,d}}$ increases exponentially with delay and the sequence of event probabilities $\{p_{_{\gamma}}\}_{^{1}}^{_{\infty}}$ diverges. For $\mu < 1$, the decreasing probability law can be checked for Lyapunov mean square stability using Theorem~\ref{Thm:SuffGeoSeries}.

\section{Example} \label{S:Sim}

We now illustrate some of the results presented in this work. We begin with an illustration of the upper bound derived in Section~\ref{SS:PDFmaj}. Our next example illustrates how the sufficient conditions for Lyapunov mean square stability, presented in Theorem~\ref{Thm:SuffGeoSeries}, can be used to infer stability properties of the control system. Our third example illustrates the design of a constant-law scheduler that guarantees Lyapunov mean square stability. The final example illustrates the selection of event thresholds corresponding to a given design.

\begin{figure*}[ptb]
\begin{center}
\includegraphics*[scale=0.4,viewport=50 0 1000 300]{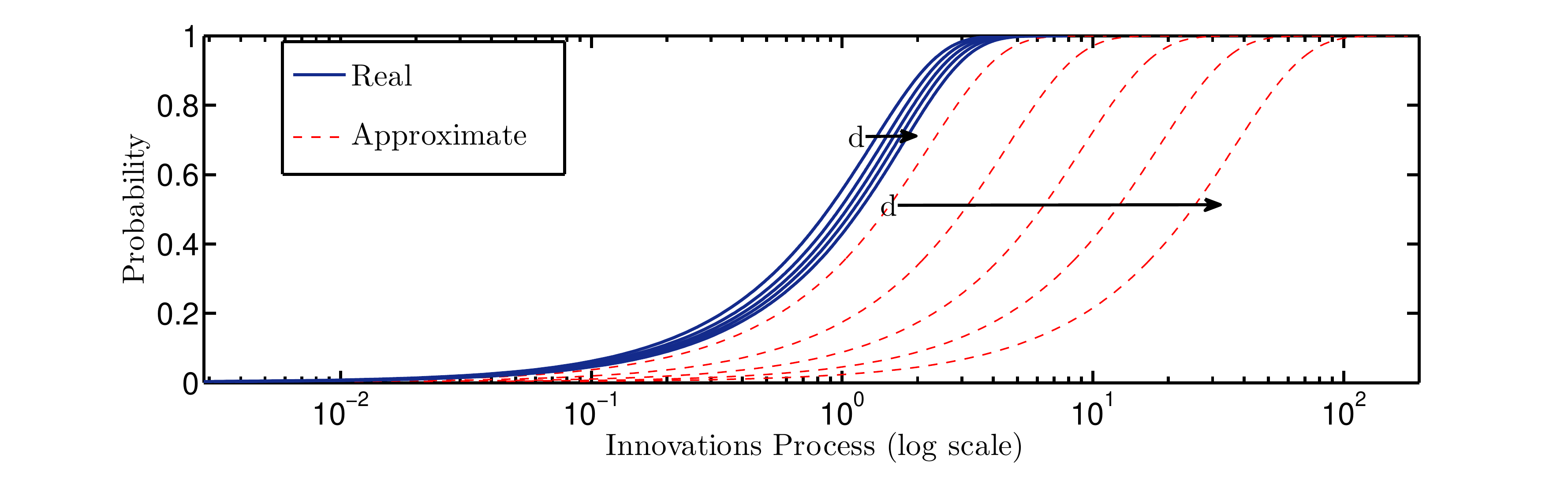}
\caption{The approximate PDF in Lemma~\ref{Lemma:HatPhi} is majorized by the actual PDF, as seen in this comparison of the CDFs. The approximate distribution has a larger variance, and is an upper bound for the actual variance. }
\label{Fig:CompCDF}
\end{center}
\end{figure*}

\begin{example}[Illustration of Majorization] \label{Ex:MajPDF}
In Lemma~\ref{Lemma:HatPhi}, we find an approximating PDF $\hat{\phi}_{_{(I,d)}}$, which is majorized by the real PDF $\phi_{_{(I,d)}}$, for all delays $d>0$. We illustrate this for a control system with parameters $A=2$, $B=1$, $\Sigma_w = 1$ and a constant event threshold $\Delta_d=1$, for all $d>0$. The CRM persistence probability is set to $p_{_{\alpha}}=1$, and the conditional probability of a busy channel is $p=0.6$. For this setup, we compare the cumulative distribution function (CDF) corresponding to $\phi_{_{(I,d)}}$, with the CDF corresponding to $\hat{\phi}_{_{(I,d)}}$, for $d=\{1,\dots,5\}$, in Fig.~\ref{Fig:CompCDF}. The arrows indicate the increasing delays. Clearly, $\phi_{_{(I,d)}} \succ \hat{\phi}_{_{(I,d)}}$, for each of the five delays, according to Definition~\ref{Def:Maj}. This figure also illustrates why the estimation error covariance of the approximated PDF is greater than that of the real PDF.
\end{example}

\begin{figure*}[ptb]
\centering
\includegraphics*[scale=0.4,viewport=0 20 1000 300]{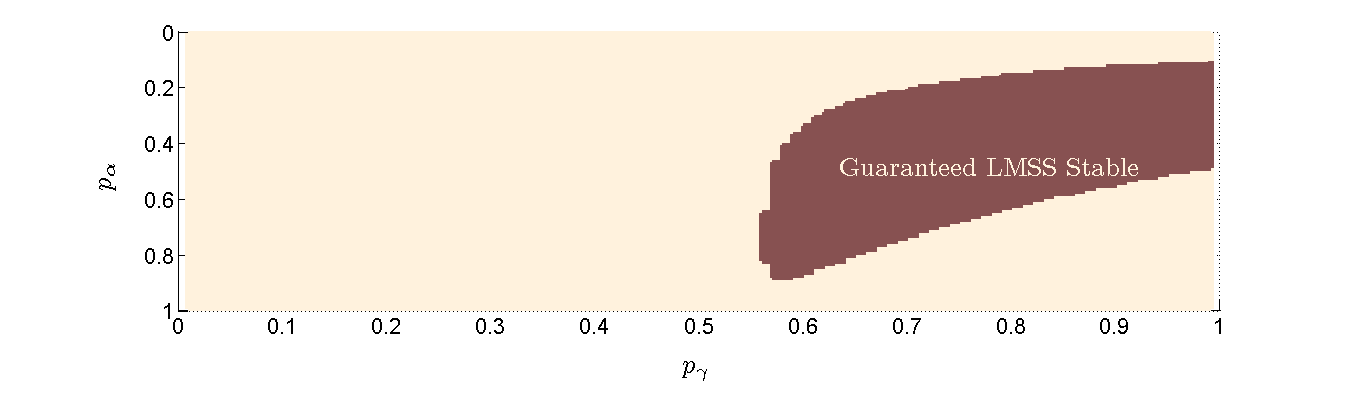}
\caption{The shaded region denotes the set of event and persistence probabilities, $p_{_{\gamma}}$ and $p_{_{\alpha}}$, respectively, that guarantee Lyapunov mean square stability. We use the sufficient conditions in Theorem~\ref{Thm:ConstLaw}, for a constant-probability scheduler, to determine Lyapunov mean square stability. }
\label{Fig:pSuccessMap}
\end{figure*}
\begin{figure*}[ptb]
\includegraphics*[scale=0.235,viewport=0 40 1000 450]{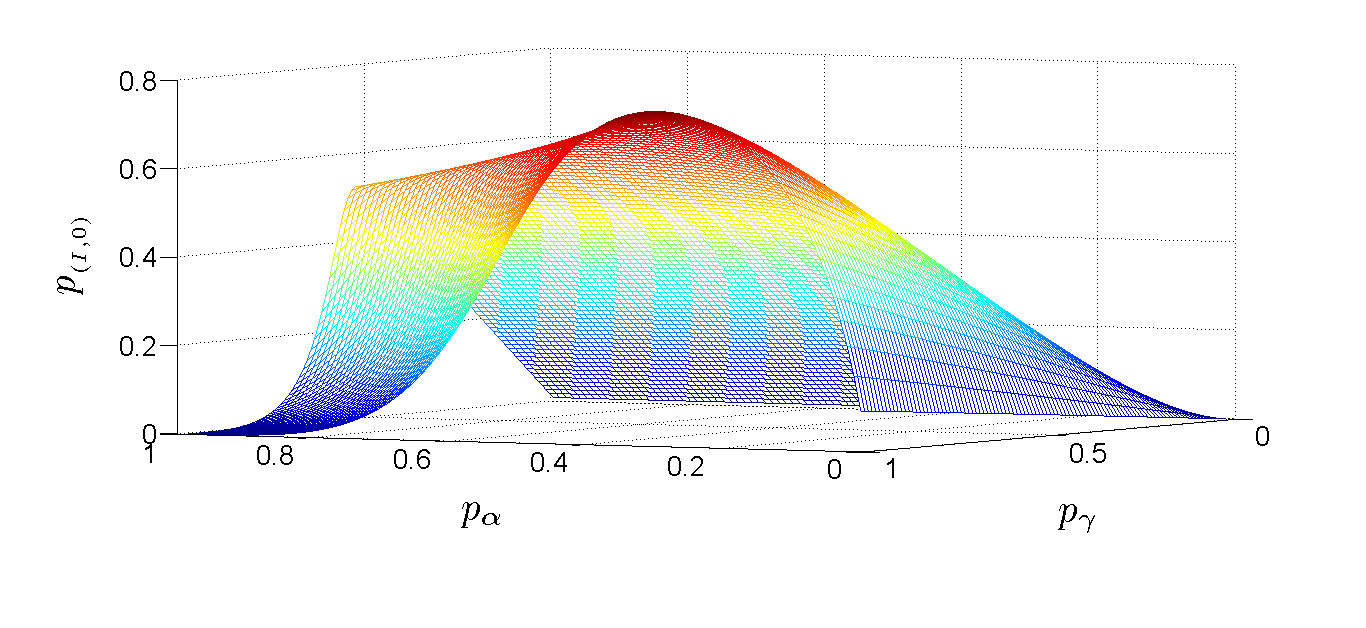} \hfill
\includegraphics*[scale=0.235,viewport=10 40 1000 450]{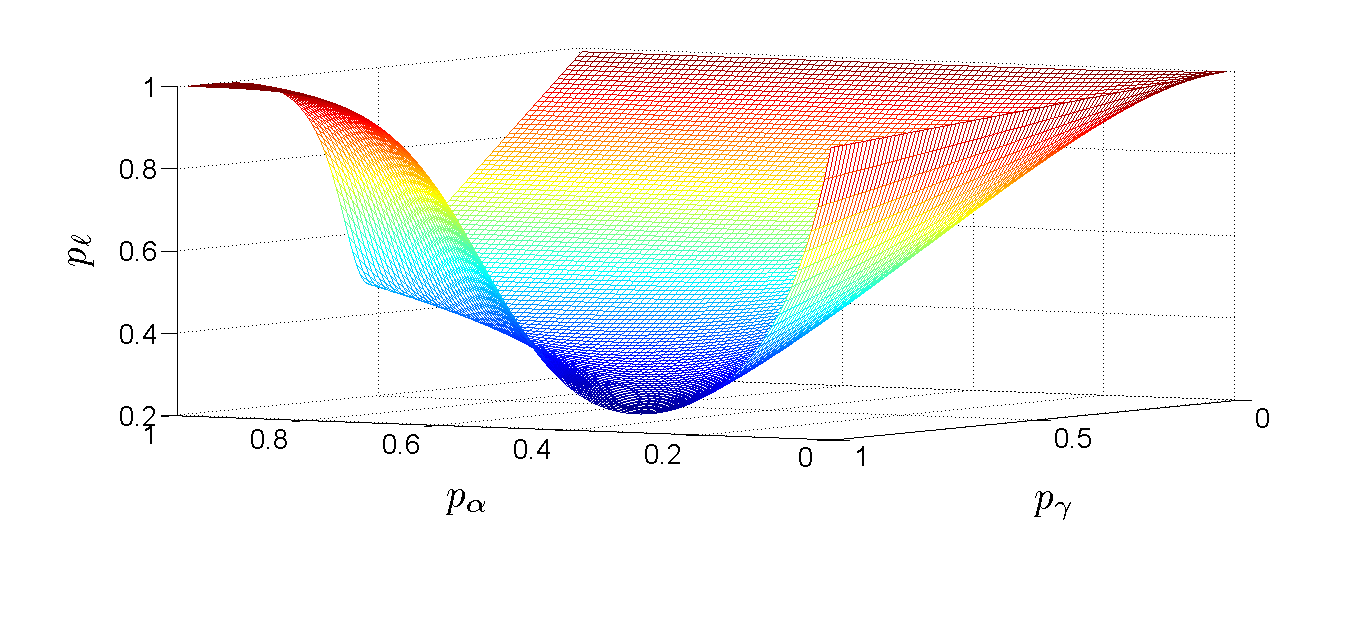}
\caption{A surface plot of the network reliability $\pi_{_{(I,0)}}$, and the probability of loss $p_{_{l}}$, respectively, versus the event probability $p_{_{\gamma}}$ and the persistence probability $p_{_{\alpha}}$, for a constant-probability scheduler. }
\label{Fig:pSuccessLossPlot}
\end{figure*}

The auxiliary system was chosen to correspond to the worst-case evolution of the real system, with saturated network traffic. Thus, the upper bound is tighter for a large busy channel probability $p$ and small state transition matrix $A$.

Next, we return to Example~\ref{Ex:LMSSvisualizer}, where we illustrated that different network sizes result in different stability properties. We use our sufficient conditions for Lyapunov mean square stability to confirm the observed stability properties for two network sizes.
\begin{example}\textbf{\upshape{(Checking for Lyapunov Mean Square Stability)}} \label{Ex:LMSSver}
We consider two network scenarios: case $1$ corresponds to a network with $M = 2$ nodes and case $2$ to a network with $M=10$ nodes. The control systems in both network scenarios are identical to the systems described in Example~\ref{Ex:LMSSvisualizer}, and so is the CRM. We use Theorem~\ref{Thm:SuffGeoSeries} to show that in case~$1$, Lyapunov mean square stability is achievable, and that in case~$2$, Lyapunov mean square stability cannot be guaranteed. This can be seen by using (\ref{Eq:SuffGeoSeries}), where we see that the idle state probabilities must achieve a ratio of less than $0.5$ for large $d$. Case $1$ achieves a ratio of less than $0.1$ for $d>10$, whereas case $2$ has a ratio of $0.98$ even after $d=50$. The Lyapunov mean square stability properties can be inferred from a trace of the state $x$ as illustrated in Fig.~\ref{Fig:MultiX}.
\end{example}

\begin{figure}[ptb]
\begin{center}
\subfigure[Probability of Transmission Success, $A = 1.25$]{\label{Fig:Comparison_Succ1}
\includegraphics*[scale=0.225,viewport=25 0 1000 430]{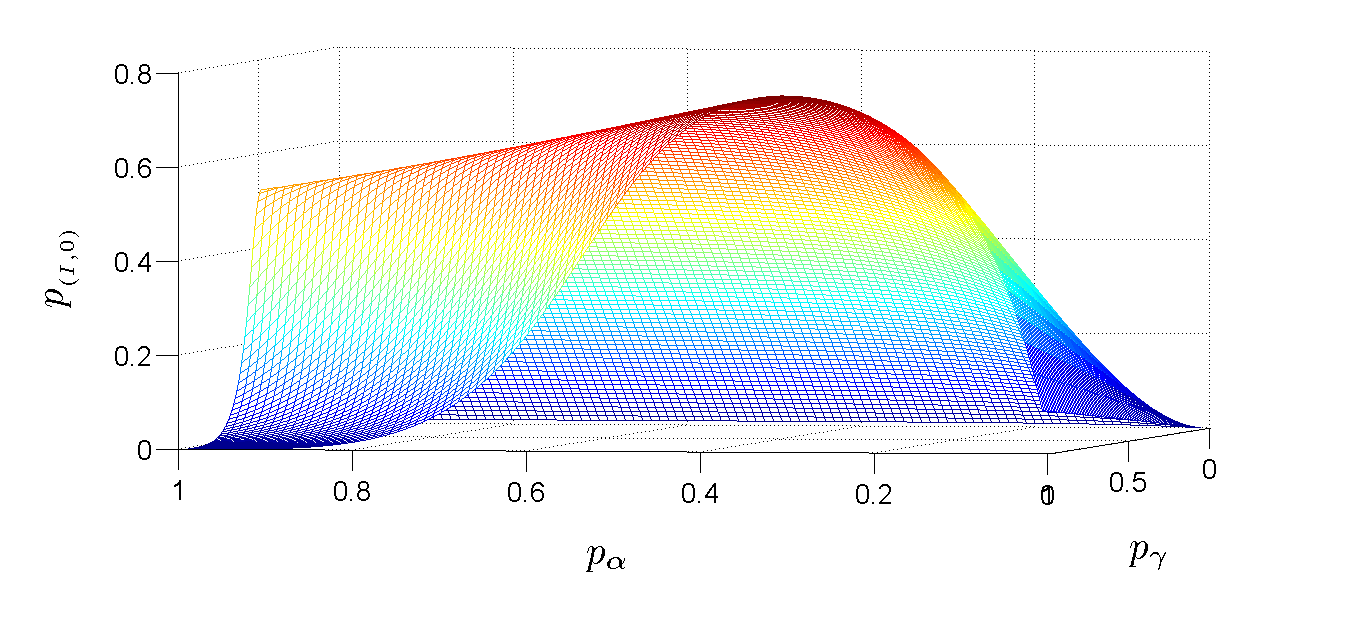}}
\subfigure[Probability of Transmission Success, $A = 2$]{\label{Fig:Comparison_Succ2}
\includegraphics*[scale=0.225,viewport=25 0 1000 430]{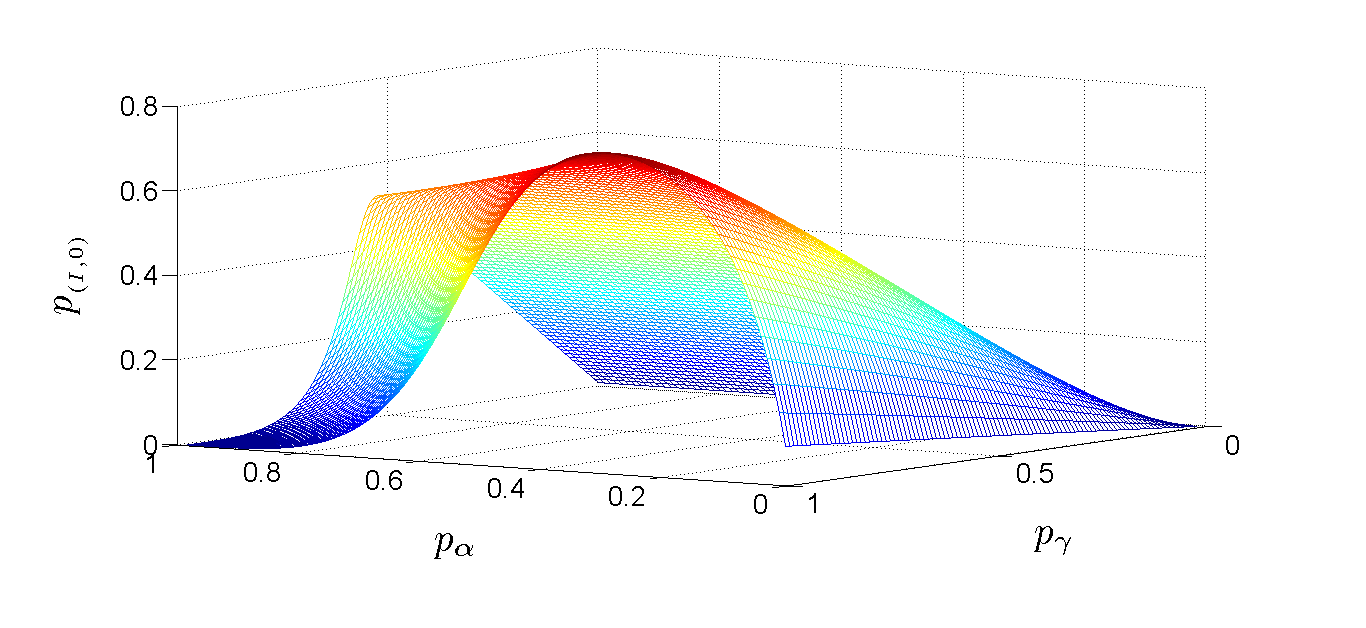}}
\subfigure[Probability of Loss, $A = 1.25$]{\label{Fig:Comparison_Loss1}
\includegraphics*[scale=0.225,viewport=25 0 1000 300]{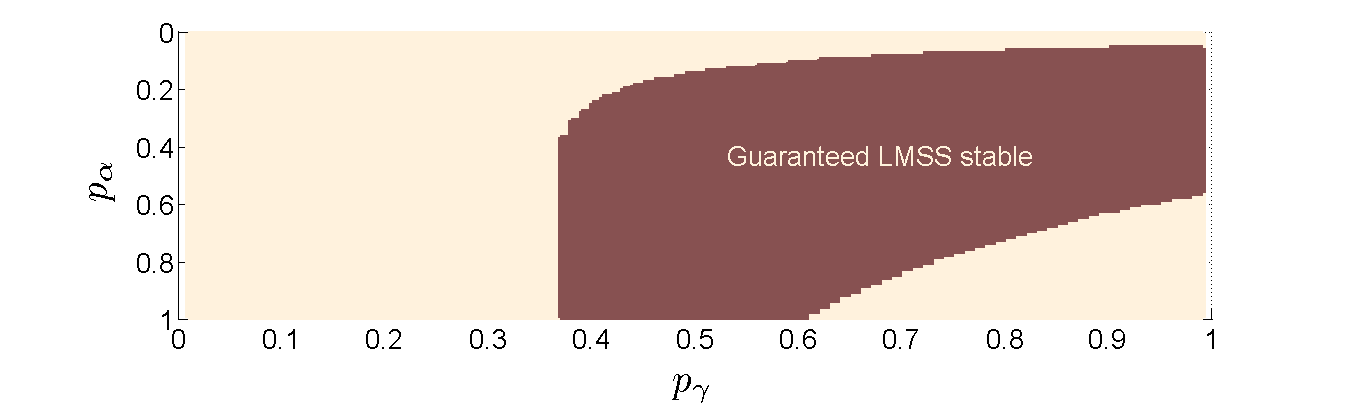}}
\subfigure[Probability of Loss, $A = 2$]{\label{Fig:Comparison_Loss2}
\includegraphics*[scale=0.24,viewport=25 25 1000 300]{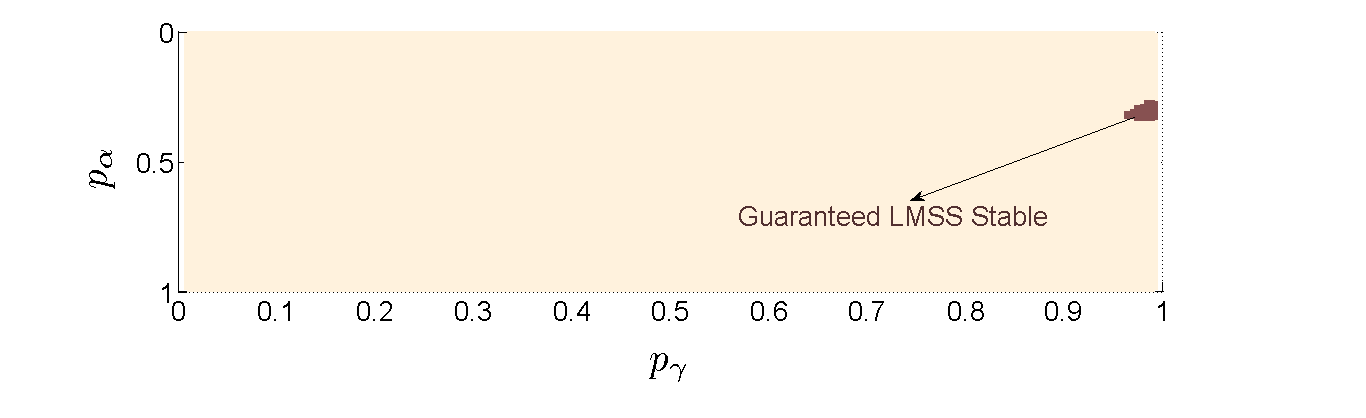}}
\caption{Surface Plots of the probability of transmission success (above), and plots indicating the region of stability (below), for $A = 1.25$ (left) and $A = 2$ (right). The region guaranteed to be Lyapunov mean square stability by our sufficient conditions shrinks considerably for highly unstable control systems.}
\end{center}
\end{figure}

In the next example, we illustrate the design procedure outlined in Fig.~\ref{Fig:FlowChart}.
\begin{example}[Constant-Probability Policy] \label{Ex:ConstLawDesign}
We consider a network of $M=5$ control systems with identical parameters $A = 1.5$, $B = 1$ and $\Sigma_w = 1$, and $r_{\max} = 10$ retransmissions in the CRM. For simplicity, we assume that the persistence probability $p_{_{\alpha}}$ is constant for all retransmission attempts. Using an algorithm similar to the one outlined in the flowchart in Fig.~\ref{Fig:FlowChart}, we obtain the set of event probabilities $p_{_{\gamma}}$, and the set of persistence probabilities $p_{_{\alpha}}$, that result in Lyapunov mean square stability. The results are depicted in Fig.~\ref{Fig:pSuccessMap}. The shaded region in the figure corresponds to this set. In Fig.~\ref{Fig:pSuccessLossPlot}, we present surface plots of the network reliability $\pi_{_{(I,0)}}$ and the probability of loss $p_{_{l}}$, respectively, versus $p_{_{\gamma}}$ and $p_{_{\alpha}}$. It is interesting to note the importance of jointly selecting $p_{_{\gamma}}$ and $p_{_{\alpha}}$.

Now, we compare the Lyapunov mean square stability regions obtained for the same network, but with control system parameters $A=1.25$ and $A=2$, i.e., less unstable and more unstable systems, respectively. The surface plots of the probability of successful transmission are shown in Fig.~\ref{Fig:Comparison_Succ1} and Fig.~\ref{Fig:Comparison_Succ2}. The shaded regions in Fig~\ref{Fig:Comparison_Loss1} and Fig~\ref{Fig:Comparison_Loss2} denote the sets of event and persistence probabilities that guarantee Lyapunov mean square stability. Notice how the Lyapunov mean square stability region given by our sufficient condition in Theorem~\ref{Thm:SuffGeoSeries} shrinks as $A$ increases.
\end{example}

\begin{figure*}[ptb] \label{Fig:SimResults}
\centering
\includegraphics*[width=20cm,viewport = 30 200 1200 450]{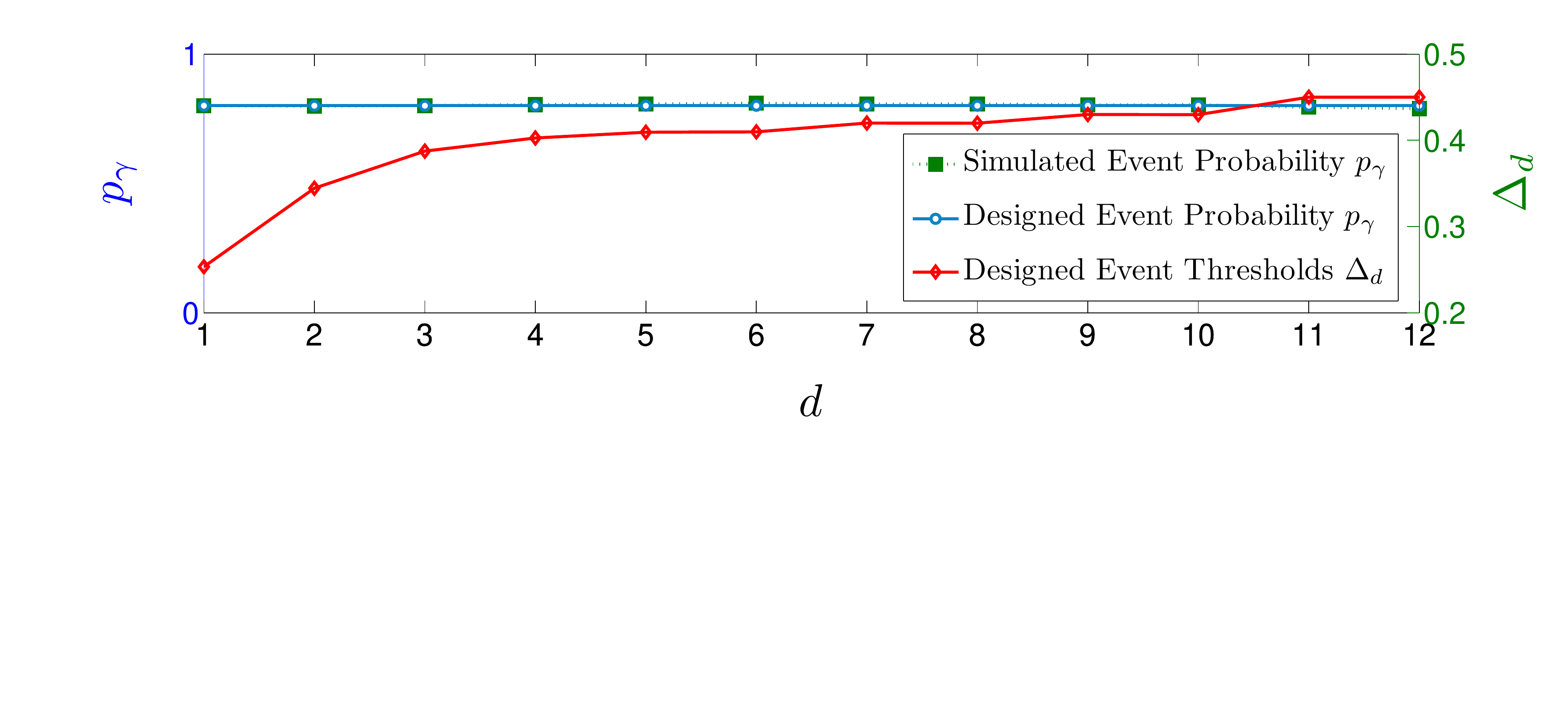}
\caption[Constant Law Policy]{A comparison of analytical and simulated values of the event probability $p_{_{\gamma}}$, plotted against the axis on the left. The numerically computed values of the event thresholds $\Delta_d$ are plotted against the axis on the right. }
\end{figure*}

We now illustrate how to select event thresholds for an event probability $p_{_{\gamma}}$ and persistence probability $p_{_{\alpha}}$ chosen from the outcome of Example~\ref{Ex:ConstLawDesign}.
\begin{example}[Selecting Event Thresholds] \label{Ex:SelectLevels}
For the control system with parameters $A = 1.5$, , $B = 1$ and $\Sigma_w = 1$, choose $p_{_{\gamma}} = 0.8$ and $p_{_{\alpha}} = 0.4$. This choice of probabilities yields a network reliability of $\pi_{_{(I,0)}} = 0.7056$ and a loss probability of $p_{_{l}} = 0.2944$, from Fig.~\ref{Fig:pSuccessMap} and Fig.~\ref{Fig:pSuccessLossPlot}, respectively. The delay distribution for the constant-probability scheduler is easily seen to be given by $\Pr(d_k = d) = \pi_{_{(I,0)}} \cdot p_{_{l}}^d$, for any delay $d \ge 0$. The exponential delay distribution considerably simplifies our task. We now need to identify only a set of $D$ event thresholds, where we choose $D$ to be sufficiently larger than the smallest probability we wish to consider. In this example, we choose $D = 12$. 

We now numerically compute event thresholds that give us the above event probabilities. To do this, we simulate the evolution of distributions described in (\ref{Eq:pdfNd})--(\ref{Eq:UpdatePDF}), and assign the event thresholds as $\Delta_d := t: \int_{|x| \le t} \phi_d dx = p_{_{\gamma}}$, for all $d \ge 0$. We present the event thresholds, thus identified, in Fig.~\ref{Fig:SimResults}. To validate our design procedure, we run Monte Carlo simulations using the thresholds identified above, and confirm that the event probabilities we obtain are as desired, as shown in Fig.~\ref{Fig:SimResults}. For delays larger than $D = 12$, the probabilities we obtain are not accurate, as there are too few instances of these events to result in a precise value. 
\end{example}

\section{Conclusions} \label{S:Conclusions}

We have presented a stability analysis and synthesis for event-triggering policies in a network of control systems that use a CRM to access the network. The event-triggering policy and the CRM sometimes result in congestion, and consequently packet losses and delays. We presented a two-step approach to design event-based systems in this network. We began by designing event probabilities that guarantee Lyapunov mean square stability for each control system in the network, and then illustrated how to select event thresholds to achieve the designed event probabilities. Our event probability designs were based on the stability analysis presented in this paper.

The results in this paper indicate a potential for performance optimization through joint design of the event probabilities and the CRM. This, and a study into the feasibility of implementing such a joint design, are left for future work.

\bibliographystyle{elsarticle-num}
\bibliography{FullList_NoUrl}


\appendix
\section{Properties of the Majorization Operator} \label{App:MajLemmas}

We first need the following result on neat and even PDFs.
\begin{lemma} \label{Lemma:SymmNIConv}
If the PDFs $\phi_a$ and $\phi_b$ on $\mathbb{R}$ are neat and even, then $\phi_a \ast \phi_b$ is also neat and even.
\end{lemma}
\begin{proof}
The PDF $\phi_b$ is a convex combination of indicator functions, $\chi_n(x) = 1$ for $x \in [-n,n]$ and zero otherwise. For any $n$, note that $\phi_a \ast \chi_n$ is symmetric and non-increasing. Convex combinations of neat distributions are neat, and hence, the result follows. \hfill \qed
\end{proof}

We now present a series of results that use the majorization operator. The proofs presented here are adapted from~\cite{Hajek2008}. These results are used in the proofs presented in Section~\ref{SS:PDFmaj}.
\begin{lemma} \label{Lemma:Maj1}
If the PDFs $\phi_a$ and $\phi_b$ on $\mathbb{R}^n$ are such that $\phi_a \succ \phi_b$, then $\int \phi_a^\sigma (x) h(x) dx \ge \int \phi_b^\sigma (x) h(x) dx$ for any symmetric non-increasing function $h$.
\end{lemma}
\begin{proof}
The function $h$ is a convex combination of indicator functions of balls centered at the origin. For any such indicator function, the above result is obvious from the definition of majorization. Hence, the result follows. \hfill \qed
\end{proof}

\begin{lemma} \label{Lemma:Maj2}
If the PDFs $\phi_a$, $\phi_b$ and $\psi$ on $\mathbb{R}^n$ are such that $\phi_a \succ \phi_b$, and if $\phi_a$ and $\psi$ are symmetric non-increasing, then $\phi_a \ast \psi \succ \phi_b \ast \psi$.
\end{lemma}
This proof uses Riesz's rearrangement inequality and is given as proof of Lemma~$6.7$ in~\cite{Hajek2008}.

\begin{lemma} \label{Lemma:Maj3}
If the PDFs $\phi_a$ and $\phi_b$ on $\mathbb{R}$ are such that $\phi_a \succ \phi_b$, and $\phi^+(x) \triangleq \frac{1}{|a|} \phi(\frac{x}{a})$, then $\phi^+_a \succ \phi^+_b$.
\end{lemma}
\begin{proof}
Using the definition of majorization, and the definitions of $\phi_a^+$ and $\phi_b^+$, the result can be shown to hold directly. \hfill \qed
\end{proof}

\begin{lemma} \label{Lemma:Maj4}
If the symmetric non-increasing PDFs $\phi_a$ and $\phi_b$ on $\mathbb{R}^n$ are such that $\phi_a \succ \phi_b$, and if $h$ be a symmetric non-decreasing positive function, then $\int \phi_a^\sigma(x) h(x) dx \le \int \phi_b^\sigma(x) h(x) dx$.
\end{lemma}
\begin{proof}
Note that the function $h$ is symmetric and quasi-concave, thus making it Schur-concave. It is known that $\E_{\phi_a}[h] \le \E_{\phi_b}[h]$, for Schur-concave functions. Thus, the desired result follows. \hfill \qed
\end{proof}

\section{Other Lemmas}

\begin{lemma} \label{Lemma:ReducedVariance}
For a general $n^{\textrm{th}}$-order plant in the network setup given by (\ref{Eq:StateSpace})--(\ref{Eq:Controller}), the posterior variance of the innovations is less than its a priori value, i.e., $\mathrm{var}(\phi^{e}_{_{(I,d+1)}}) \le \mathrm{var}(\phi_{_{(I,d)}})$.
\end{lemma}
\begin{proof}
We can find expressions for the a priori variance, denoted $\sigma^2_{_{(I,d)}} \triangleq \mathrm{var}(\phi_{_{(I,d)}})$, and the posterior variance, denoted  $(\sigma^{e}_{_{(I,d+1)}})^2 \triangleq \mathrm{var}(\phi^e_{_{(I,d+1)}})$, as
\begin{align*}
\sigma^2_{_{(I,d)}} &= \sigma^2_{\Delta^-,d} + \sigma^2_{\Delta^+,d} \; , \\
(\sigma^{e}_{_{(I,d+1)}})^2 &= \sigma^2_{\Delta^-,d} \frac{1}{1 - p_{_{\gamma,d}} p_{_{\alpha}} q} + \sigma^2_{\Delta^+,d} \frac{(1 - p_{_{\alpha}} q)}{1 - p_{_{\gamma,d}} p_{_{\alpha}} q} \; ,
\end{align*}
where $\sigma^2_{\Delta^-,d} = \int_{|\tilde{x}| \le \Delta_{d}} |\tilde{x}|^2 \psi(\tilde{x}) d \tilde{x}$ and $\sigma^2_{\Delta^+,d} = \int_{|\tilde{x}| > \Delta_{d}} |\tilde{x}|^2 \psi(\tilde{x}) d \tilde{x}$. Now, the variance of the posterior distribution can be rewritten as
\begin{align*}
(\sigma^{e}_{_{(I,d+1)}})^2 &= \sigma^2_{_{(I,d)}} + \sigma^2_{\Delta^-,d} \left( \frac{1}{1 - p_{_{\gamma,d}} p_{_{\alpha}} q} - 1 \right) \\
 &\quad \quad \quad \quad \quad + \sigma^2_{\Delta^+,d} \left( \frac{(1 - p_{_{\alpha}} q)}{1 - p_{_{\gamma,d}} p_{_{\alpha}} q} - 1\right) \\
&\le \sigma^2_{_{(I,d)}} + \max \bigg( \sigma^2_{\Delta^-,d} \left( \frac{1}{1 - p_{_{\gamma,d}} p_{_{\alpha}} q} - 1 \right)  \\
& \quad \quad \quad \quad \quad + \sigma^2_{\Delta^+,d} \left( \frac{(1 - p_{_{\alpha}} q)}{1 - p_{_{\gamma,d}} p_{_{\alpha}} q} - 1\right) \bigg) \; .
\end{align*}
The maximum value of the first term can be found by evaluating the integral at the upper boundary to obtain $\max \sigma^2_{\Delta^-,d} = \Delta_d^2 q_{_{\gamma,d}}$. However, the second term is negative as $\frac{(1 - p_{_{\alpha}} q)}{1 - p_{_{\gamma,d}} p_{_{\alpha}} q} < 1$. The maximum value of this term is found by evaluating the integral at the lower boundary. Doing so, we obtain
\begin{align*}
(\sigma^{e}_{_{(I,d+1)}})^2 &\le \sigma^2_{_{(I,d)}} + \Delta_d^2 \bigg( q_{_{\gamma,d}} \left( \frac{1}{1 - p_{_{\gamma,d}} p_{_{\alpha}} q} - 1 \right) \\
&\quad \quad \quad \quad \quad - p_{_{\gamma,d}} \left( \frac{(1 - p_{_{\alpha}} q)}{1 - p_{_{\gamma,d}} p_{_{\alpha}} q} - 1\right) \bigg) \\
&= \sigma^2_{_{(I,d)}} + 0 \; ,
\end{align*}
where it is easy to check that the terms in the inner bracket sum to zero. \hfill \qed
\end{proof}

%

\end{document}

%% file: crenv.tex
\usepackage{amsmath} 
\usepackage{amssymb}  
\usepackage{bm}  
\usepackage{dsfont}
\usepackage{graphicx,color} 
\usepackage{epsfig} 
\usepackage{epstopdf}
\usepackage{color}
\usepackage[tight,footnotesize]{subfigure}
\usepackage{cite}
\usepackage{url}
\usepackage{array}
\usepackage[svgnames,table]{xcolor}
\usepackage{dpfloat, booktabs}
\usepackage{pdflscape}
\usepackage{rotating}
\usepackage{longtable}
\usepackage{booktabs}
\usepackage{multirow}
\usepackage{algpseudocode}
\usepackage{algorithm}

\DeclareMathOperator*{\E}{{\mathbf E}}        
\let\Pr\undefined 
\DeclareMathOperator{\Pr}{{\mathbf P}}        

\DeclareMathOperator{\tr}{tr}

\newtheorem{theorem}{Theorem}[section]

\newtheorem{lemma}[theorem]{Lemma}
\newtheorem{proposition}[theorem]{Proposition}

\newtheorem{definition}{Definition}[section]

\newtheorem{assumption}{Assumption}[section]
\newtheorem{example}{Example}[section]

\newtheorem{remark}{Remark}[section]

\newenvironment{proof}[1][Proof]{\begin{trivlist}
\item[\hskip \labelsep {\bfseries #1}]}{\end{trivlist}}
\renewcommand{\qed}{$\blacksquare$}


%% file: figuresTR/macNCS2.eps_tex

\begingroup
  \makeatletter
  \providecommand\color[2][]{%
    \errmessage{(Inkscape) Color is used for the text in Inkscape, but the package 'color.sty' is not loaded}
    \renewcommand\color[2][]{}%
  }
  \providecommand\transparent[1]{%
    \errmessage{(Inkscape) Transparency is used (non-zero) for the text in Inkscape, but the package 'transparent.sty' is not loaded}
    \renewcommand\transparent[1]{}%
  }
  \providecommand\rotatebox[2]{#2}
  \ifx\svgwidth\undefined
    \setlength{\unitlength}{321.4953125pt}
  \else
    \setlength{\unitlength}{\svgwidth}
  \fi
  \global\let\svgwidth\undefined
  \makeatother
  \begin{picture}(1,0.45155139)%
    \put(0,0){\includegraphics[width=\unitlength]{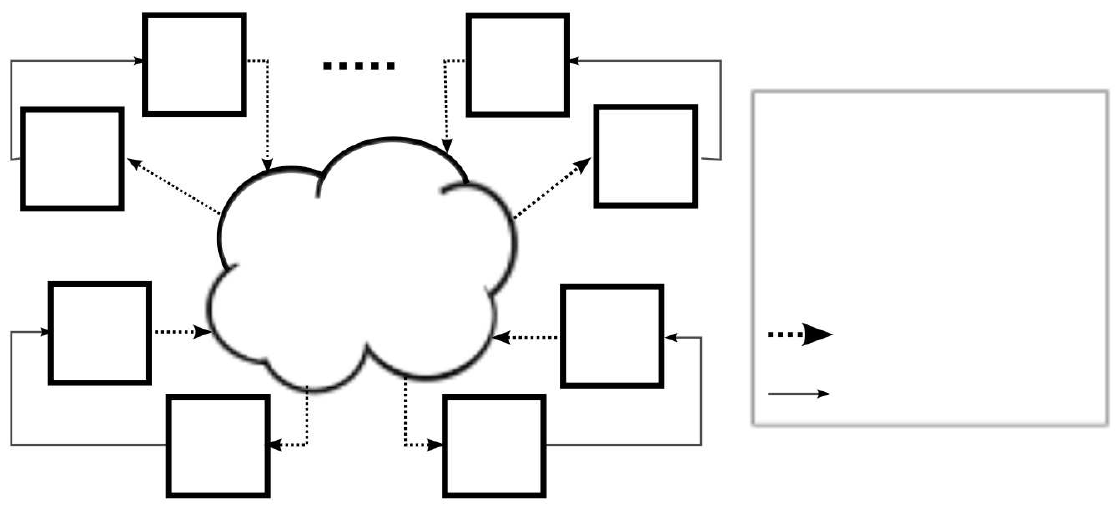}}%
    \put(0.28661729,0.25){\color[rgb]{0,0,0}\makebox(0,0)[lt]{\resizebox{1cm}{!}{$\mathcal{N}$}}}%
    \put(0.775,0.3189232){\color[rgb]{0,0,0}\makebox(0,0)[lb]{\smash{$j^\textrm{th}$ Plant}}}%
    \put(0.775,0.26602084){\color[rgb]{0,0,0}\makebox(0,0)[lb]{\smash{$j^\textrm{th}$ Controller}}}%
    \put(0.775,0.1452268){\color[rgb]{0,0,0}\makebox(0,0)[lb]{\smash{Sensor Link}}}%
    \put(0.775,0.09232444){\color[rgb]{0,0,0}\makebox(0,0)[lb]{\smash{Actuator Link}}}%
    \put(0.775,0.19812916){\color[rgb]{0,0,0}\makebox(0,0)[lb]{\smash{Network}}}%
    \put(0.07,0.145){\color[rgb]{0,0,0}\makebox(0,0)[lb]{\smash{$\mathcal{P}_1$}}}%
    \put(0.15,0.385){\color[rgb]{0,0,0}\makebox(0,0)[lb]{\smash{$\mathcal{P}_2$}}}%
    \put(0.1784,0.04){\color[rgb]{0,0,0}\makebox(0,0)[lb]{\smash{$\mathcal{C}_1$}}}%
    \put(0.045,0.3){\color[rgb]{0,0,0}\makebox(0,0)[lb]{\smash{$\mathcal{C}_2$}}}%
    \put(0.69358965,0.31778685){\color[rgb]{0,0,0}\makebox(0,0)[lb]{\smash{$\mathcal{P}_j$}}}%
    \put(0.69358965,0.26488449){\color[rgb]{0,0,0}\makebox(0,0)[lb]{\smash{$\mathcal{C}_j$}}}%
    \put(0.68565431,0.19699296){\color[rgb]{0,0,0}\makebox(0,0)[lb]{\smash{$\mathcal{N}$}}}%
    \put(0.4225,0.385){\color[rgb]{0,0,0}\makebox(0,0)[lb]{\resizebox{9mm}{!}{$\mathcal{P}_{M-1}$}}}%
    \put(0.52,0.14){\color[rgb]{0,0,0}\makebox(0,0)[lb]{\smash{$\mathcal{P}_M$}}}%
    \put(0.54,0.295){\color[rgb]{0,0,0}\makebox(0,0)[lb]{\resizebox{9mm}{!}{$\mathcal{C}_{M-1}$}}}%
    \put(0.42,0.037){\color[rgb]{0,0,0}\makebox(0,0)[lb]{\smash{$\mathcal{C}_M$}}}%
  \end{picture}%
\endgroup

%% file: figuresTR/SystemModel2.eps_tex

\begingroup
  \makeatletter
  \providecommand\color[2][]{%
    \errmessage{(Inkscape) Color is used for the text in Inkscape, but the package 'color.sty' is not loaded}
    \renewcommand\color[2][]{}%
  }
  \providecommand\transparent[1]{%
    \errmessage{(Inkscape) Transparency is used (non-zero) for the text in Inkscape, but the package 'transparent.sty' is not loaded}
    \renewcommand\transparent[1]{}%
  }
  \providecommand\rotatebox[2]{#2}
  \ifx\svgwidth\undefined
    \setlength{\unitlength}{479.39765625pt}
  \else
    \setlength{\unitlength}{\svgwidth}
  \fi
  \global\let\svgwidth\undefined
  \makeatother
  \begin{picture}(1,0.40047623)%
    \put(0,0){\includegraphics[width=\unitlength]{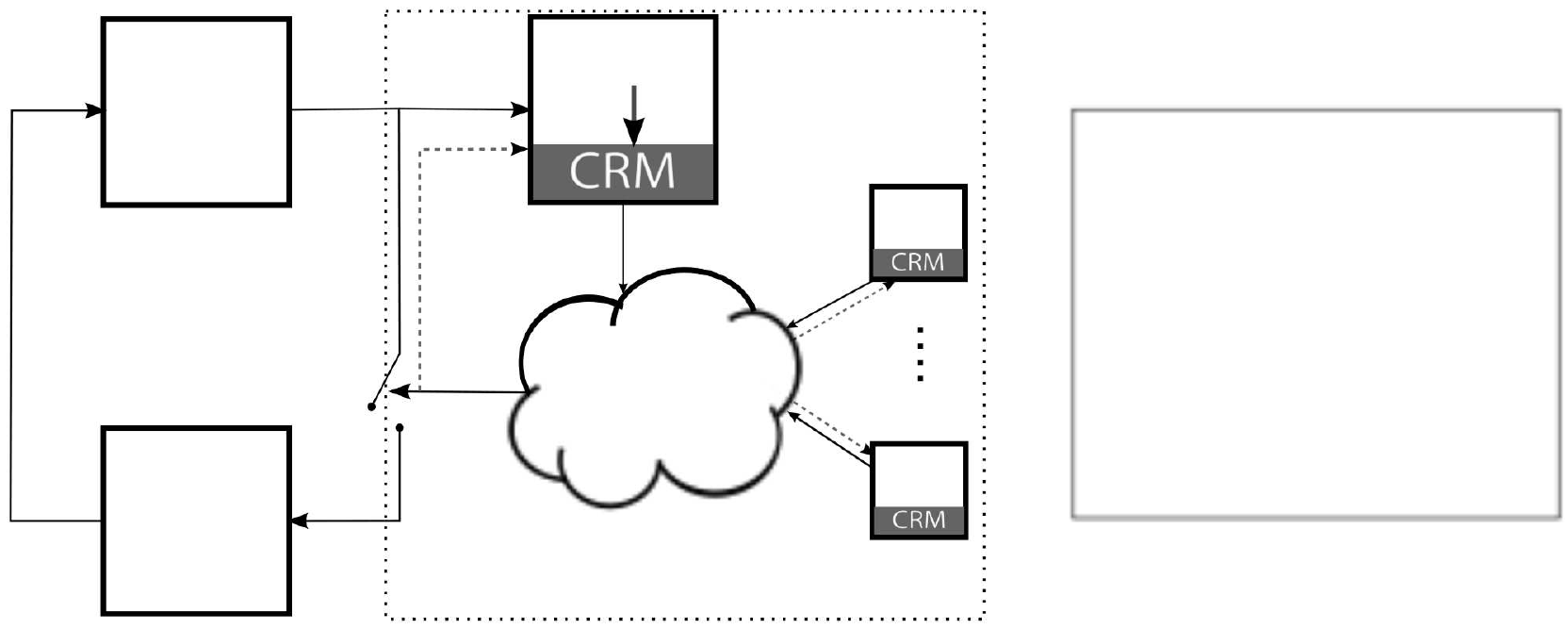}}%
    \put(0.095,0.35){\color[rgb]{0,0,0}\makebox(0,0)[lt]{\resizebox{9mm}{!}{$\mathcal{P}_{j}$}}}%
    \put(0.35,0.38){\color[rgb]{0,0,0}\makebox(0,0)[lt]{\resizebox{9mm}{!}{$\mathcal{S}_{j}$}}}%
    \put(0.095,0.09){\color[rgb]{0,0,0}\makebox(0,0)[lt]{\resizebox{9mm}{!}{$\mathcal{C}_{j}$}}}%
    \put(0.57,0.27){\color[rgb]{0,0,0}\makebox(0,0)[lt]{\resizebox{5mm}{!}{$\mathcal{S}_{1}$}}}%
    \put(0.29,0.25){\color[rgb]{0,0,0}\rotatebox{90}{\makebox(0,0)[lb]{\smash{ACK}}}}%
    \put(0.39,0.17){\color[rgb]{0,0,0}\makebox(0,0)[lt]{\resizebox{9mm}{!}{$\mathcal{N}$}}}%
    \put(0.565,0.105){\color[rgb]{0,0,0}\makebox(0,0)[lt]{\resizebox{6.5mm}{!}{$\mathcal{S}_{M}$}}}%
    \put(0.767,0.297){\color[rgb]{0,0,0}\makebox(0,0)[lb]{\smash{$j^{\textrm{th}}$ Plant}}}%
    \put(0.767,0.262){\color[rgb]{0,0,0}\makebox(0,0)[lb]{\smash{$j^{\textrm{th}}$ Controller}}}%
    \put(0.767,0.16){\color[rgb]{0,0,0}\makebox(0,0)[lb]{\smash{Network}}}%
    \put(0.706,0.297){\color[rgb]{0,0,0}\makebox(0,0)[lb]{\smash{$\mathcal{P}_{j}$}}}%
    \put(0.706,0.262){\color[rgb]{0,0,0}\makebox(0,0)[lb]{\smash{$\mathcal{C}_{j}$}}}%
    \put(0.706,0.16){\color[rgb]{0,0,0}\makebox(0,0)[lb]{\smash{$\mathcal{N}$}}}%
    \put(0.706,0.226){\color[rgb]{0,0,0}\makebox(0,0)[lb]{\smash{$\mathcal{S}_{j}$}}}%
    \put(0.767,0.226){\color[rgb]{0,0,0}\makebox(0,0)[lb]{\smash{$j^{\textrm{th}}$ State-based}}}%
    \put(0.81,0.195){\color[rgb]{0,0,0}\makebox(0,0)[lb]{\smash{Scheduler}}}%
    \put(0.69,0.125){\color[rgb]{0,0,0}\makebox(0,0)[lb]{\smash{CRM}}}%
    \put(0.767,0.125){\color[rgb]{0,0,0}\makebox(0,0)[lb]{\smash{Contention Resolution}}}%
    \put(0.81,0.095){\color[rgb]{0,0,0}\makebox(0,0)[lb]{\smash{Mechanism}}}%
    \put(0.01175,0.34){\color[rgb]{0,0,0}\makebox(0,0)[lb]{\resizebox{5mm}{!}{$u^{j}_k$}}}%
    \put(0.2,0.34){\color[rgb]{0,0,0}\makebox(0,0)[lb]{\resizebox{5mm}{!}{$x^{j}_k$}}}%
    \put(0.2,0.08){\color[rgb]{0,0,0}\makebox(0,0)[lb]{\resizebox{5mm}{!}{$y^{j}_k$}}}%
    \put(0.275,0.16){\color[rgb]{0,0,0}\makebox(0,0)[lb]{\resizebox{5mm}{!}{$\delta^{j}_k$}}}%
    \put(0.34,0.22){\color[rgb]{0,0,0}\makebox(0,0)[lb]{\resizebox{8mm}{!}{$\alpha^{j}_{k,r}$}}}%
    \put(0.495,0.215){\color[rgb]{0,0,0}\makebox(0,0)[lb]{\resizebox{8mm}{!}{$\alpha^{1}_{k,r}$}}}%
    \put(0.495,0.071){\color[rgb]{0,0,0}\makebox(0,0)[lb]{\resizebox{8mm}{!}{$\alpha^{M}_{k,r}$}}}%
    \put(0.42,0.319){\color[rgb]{0,0,0}\makebox(0,0)[lb]{\resizebox{5mm}{!}{$\gamma^{j}_k$}}}%
  \end{picture}%
\endgroup

%% file: figuresTR/StateBasedDualPred2.eps_tex

\begingroup
  \makeatletter
  \providecommand\color[2][]{%
    \errmessage{(Inkscape) Color is used for the text in Inkscape, but the package 'color.sty' is not loaded}
    \renewcommand\color[2][]{}%
  }
  \providecommand\transparent[1]{%
    \errmessage{(Inkscape) Transparency is used (non-zero) for the text in Inkscape, but the package 'transparent.sty' is not loaded}
    \renewcommand\transparent[1]{}%
  }
  \providecommand\rotatebox[2]{#2}
  \ifx\svgwidth\undefined
    \setlength{\unitlength}{430.07734375pt}
  \else
    \setlength{\unitlength}{\svgwidth}
  \fi
  \global\let\svgwidth\undefined
  \makeatother
  \begin{picture}(1,0.4257795)%
    \put(0,0){\includegraphics[width=\unitlength]{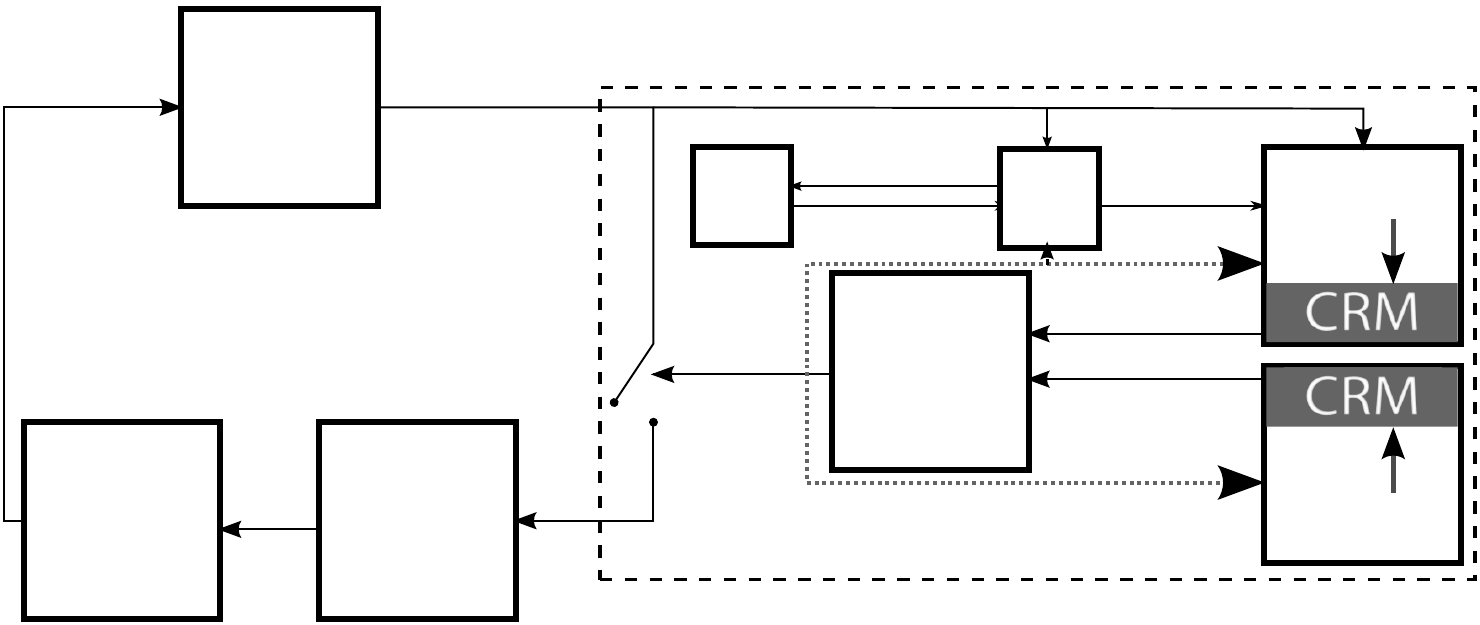}}%
    \put(0.16,0.38){\color[rgb]{0,0,0}\makebox(0,0)[lt]{\resizebox{1cm}{!}{$\mathcal{P}_{j}$}}}%
    \put(0.245,0.09){\color[rgb]{0,0,0}\makebox(0,0)[lt]{\resizebox{1.2cm}{!}{$\mathcal{O}_{j}$}}}%
    \put(0.595,0.2){\color[rgb]{0,0,0}\makebox(0,0)[lt]{\resizebox{1cm}{!}{$\mathcal{R}$}}}%
    \put(0.865,0.305){\color[rgb]{0,0,0}\makebox(0,0)[lt]{\resizebox{1cm}{!}{$\mathcal{S}_{j}$}}}%
    \put(0.865,0.105){\color[rgb]{0,0,0}\makebox(0,0)[lt]{\resizebox{1cm}{!}{$\mathcal{N}_{j}$}}}%
    \put(0.54,0.191){\color[rgb]{0.29411765,0.29411765,0.29411765}\rotatebox{90}{\makebox(0,0)[lb]{\smash{ACK}}}}%
    \put(0.055,0.095){\color[rgb]{0,0,0}\makebox(0,0)[lt]{\resizebox{1cm}{!}{$\mathcal{C}_{j}$}}}%
    \put(0.76,0.07){\color[rgb]{0.29411765,0.29411765,0.29411765}\makebox(0,0)[lb]{\smash{ACK}}}%
    \put(0.4825,0.305){\color[rgb]{0,0,0}\makebox(0,0)[lt]{\resizebox{6mm}{!}{$\mathcal{C}_{j}$}}}%
    \put(0.685,0.305){\color[rgb]{0,0,0}\makebox(0,0)[lt]{\resizebox{8mm}{!}{$\mathcal{O}_{j}$}}}%
    \put(0.05,0.36){\color[rgb]{0,0,0}\makebox(0,0)[lb]{\resizebox{5mm}{!}{$u^{j}_k$}}}%
    \put(0.28,0.36){\color[rgb]{0,0,0}\makebox(0,0)[lb]{\resizebox{5mm}{!}{$x^{j}_k$}}}%
    \put(0.16,0.075){\color[rgb]{0,0,0}\makebox(0,0)[lb]{\resizebox{8mm}{!}{$\hat{x}^{c,j}_{^{k|k}}$}}}%
    \put(0.36,0.08){\color[rgb]{0,0,0}\makebox(0,0)[lb]{\resizebox{5mm}{!}{$y^{j}_k$}}}%
    \put(0.46,0.18){\color[rgb]{0,0,0}\makebox(0,0)[lb]{\resizebox{5mm}{!}{$\delta^{j}_k$}}}%
    \put(0.565,0.3){\color[rgb]{0,0,0}\makebox(0,0)[lb]{\resizebox{12mm}{!}{$\hat{x}^{c,j}_{^{k-1|k-1}}$}}}%
    \put(0.575,0.245){\color[rgb]{0,0,0}\makebox(0,0)[lb]{\resizebox{8mm}{!}{$u^{j}_{k-1}$}}}%
    \put(0.765,0.29){\color[rgb]{0,0,0}\makebox(0,0)[lb]{\resizebox{12mm}{!}{$\hat{x}^{s,j}_{^{k|\tau_{k-1}}}$}}}%
    \put(0.745,0.20){\color[rgb]{0,0,0}\makebox(0,0)[lb]{\resizebox{9mm}{!}{$\alpha^{j}_{k,r}$}}}%
    \put(0.745,0.12){\color[rgb]{0,0,0}\makebox(0,0)[lb]{\resizebox{9mm}{!}{$\alpha^{N,j}_{k,r}$}}}%
    \put(0.945,0.08){\color[rgb]{0,0,0}\makebox(0,0)[lb]{\resizebox{5mm}{!}{$n^{j}_k$}}}%
    \put(0.945,0.25){\color[rgb]{0,0,0}\makebox(0,0)[lb]{\resizebox{5mm}{!}{$\gamma^{j}_k$}}}%
  \end{picture}%
\endgroup

%% file: figuresTR/TimeLines2.eps_tex

\begingroup
  \makeatletter
  \providecommand\color[2][]{%
    \errmessage{(Inkscape) Color is used for the text in Inkscape, but the package 'color.sty' is not loaded}
    \renewcommand\color[2][]{}%
  }
  \providecommand\transparent[1]{%
    \errmessage{(Inkscape) Transparency is used (non-zero) for the text in Inkscape, but the package 'transparent.sty' is not loaded}
    \renewcommand\transparent[1]{}%
  }
  \providecommand\rotatebox[2]{#2}
  \ifx\svgwidth\undefined
    \setlength{\unitlength}{362.12607422pt}
  \else
    \setlength{\unitlength}{\svgwidth}
  \fi
  \global\let\svgwidth\undefined
  \makeatother
  \begin{picture}(1,0.46044994)%
    \put(0,0){\includegraphics[width=\unitlength]{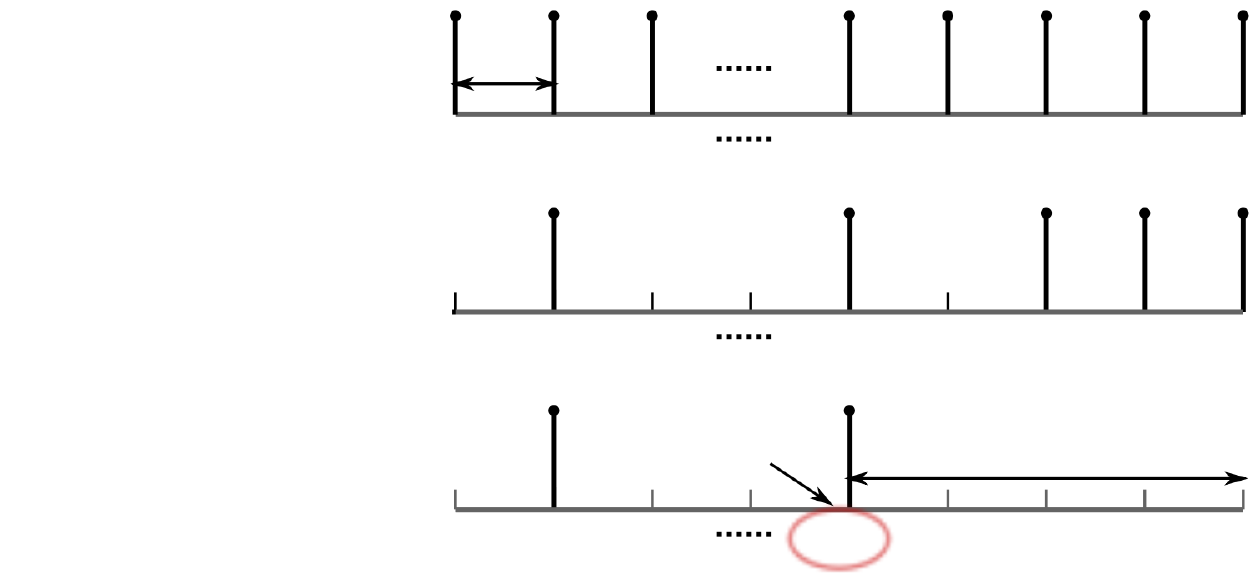}}%
    \put(0,0.43){\color[rgb]{0,0,0}\makebox(0,0)[lb]{\smash{Samples at}}}%
    \put(0.0,0.37){\color[rgb]{0,0,0}\makebox(0,0)[lb]{\smash{Plant: }}}%
    \put(0,0.27){\color[rgb]{0,0,0}\makebox(0,0)[lb]{\smash{Samples at}}}%
    \put(0,0.21){\color[rgb]{0,0,0}\makebox(0,0)[lb]{\smash{Scheduler: $\gamma_k$}}}%
    \put(0,0.11){\color[rgb]{0,0,0}\makebox(0,0)[lb]{\smash{Samples at}}}%
    \put(0,0.05){\color[rgb]{0,0,0}\makebox(0,0)[lb]{\smash{Controller: $\delta_k$}}}%
    \put(0.393,0.41){\color[rgb]{0,0,0}\makebox(0,0)[lb]{\smash{$T$}}}%
    \put(0.355,0.345){\color[rgb]{0,0,0}\makebox(0,0)[lb]{\smash{$_0$}}}%
    \put(0.435,0.345){\color[rgb]{0,0,0}\makebox(0,0)[lb]{\smash{$_1$}}}%
    \put(0.975,0.345){\color[rgb]{0,0,0}\makebox(0,0)[lb]{\smash{$_k$}}}%
    \put(0.515,0.345){\color[rgb]{0,0,0}\makebox(0,0)[lb]{\smash{$_2$}}}%
    \put(0.645,0.345){\color[rgb]{0,0,0}\makebox(0,0)[lb]{\smash{$_{k-4}$}}}%
    \put(0.725,0.345){\color[rgb]{0,0,0}\makebox(0,0)[lb]{\smash{$_{k-3}$}}}%
    \put(0.805,0.345){\color[rgb]{0,0,0}\makebox(0,0)[lb]{\smash{$_{k-2}$}}}%
    \put(0.885,0.345){\color[rgb]{0,0,0}\makebox(0,0)[lb]{\smash{$_{k-1}$}}}%
    \put(0.355,0.185){\color[rgb]{0,0,0}\makebox(0,0)[lb]{\smash{$_0$}}}%
    \put(0.435,0.185){\color[rgb]{0,0,0}\makebox(0,0)[lb]{\smash{$_1$}}}%
    \put(0.975,0.185){\color[rgb]{0,0,0}\makebox(0,0)[lb]{\smash{$_k$}}}%
    \put(0.515,0.185){\color[rgb]{0,0,0}\makebox(0,0)[lb]{\smash{$_2$}}}%
    \put(0.645,0.185){\color[rgb]{0,0,0}\makebox(0,0)[lb]{\smash{$_{k-4}$}}}%
    \put(0.725,0.185){\color[rgb]{0,0,0}\makebox(0,0)[lb]{\smash{$_{k-3}$}}}%
    \put(0.805,0.185){\color[rgb]{0,0,0}\makebox(0,0)[lb]{\smash{$_{k-2}$}}}%
    \put(0.885,0.185){\color[rgb]{0,0,0}\makebox(0,0)[lb]{\smash{$_{k-1}$}}}%
    \put(0.355,0.03){\color[rgb]{0,0,0}\makebox(0,0)[lb]{\smash{$_0$}}}%
    \put(0.435,0.03){\color[rgb]{0,0,0}\makebox(0,0)[lb]{\smash{$_1$}}}%
    \put(0.975,0.03){\color[rgb]{0,0,0}\makebox(0,0)[lb]{\smash{$_k$}}}%
    \put(0.515,0.03){\color[rgb]{0,0,0}\makebox(0,0)[lb]{\smash{$_2$}}}%
    \put(0.64,0.03){\color[rgb]{0,0,0}\makebox(0,0)[lb]{\smash{$_{k-4}$}}}%
    \put(0.725,0.03){\color[rgb]{0,0,0}\makebox(0,0)[lb]{\smash{$_{k-3}$}}}%
    \put(0.805,0.03){\color[rgb]{0,0,0}\makebox(0,0)[lb]{\smash{$_{k-2}$}}}%
    \put(0.885,0.03){\color[rgb]{0,0,0}\makebox(0,0)[lb]{\smash{$_{k-1}$}}}%
    \put(0.585,0.105){\color[rgb]{0,0,0}\makebox(0,0)[lb]{\smash{$\tau_k$}}}%
    \put(0.818,0.095){\color[rgb]{0,0,0}\makebox(0,0)[lb]{\smash{$d_k$}}}%
  \end{picture}%
\endgroup

%% file: figuresTR/MC_etDualPred_crm.eps_tex

\begingroup
  \makeatletter
  \providecommand\color[2][]{%
    \errmessage{(Inkscape) Color is used for the text in Inkscape, but the package 'color.sty' is not loaded}
    \renewcommand\color[2][]{}%
  }
  \providecommand\transparent[1]{%
    \errmessage{(Inkscape) Transparency is used (non-zero) for the text in Inkscape, but the package 'transparent.sty' is not loaded}
    \renewcommand\transparent[1]{}%
  }
  \providecommand\rotatebox[2]{#2}
  \ifx\svgwidth\undefined
    \setlength{\unitlength}{449.92768555pt}
  \else
    \setlength{\unitlength}{\svgwidth}
  \fi
  \global\let\svgwidth\undefined
  \makeatother
  \begin{picture}(1,0.50289749)%
    \put(0,0){\includegraphics[width=\unitlength]{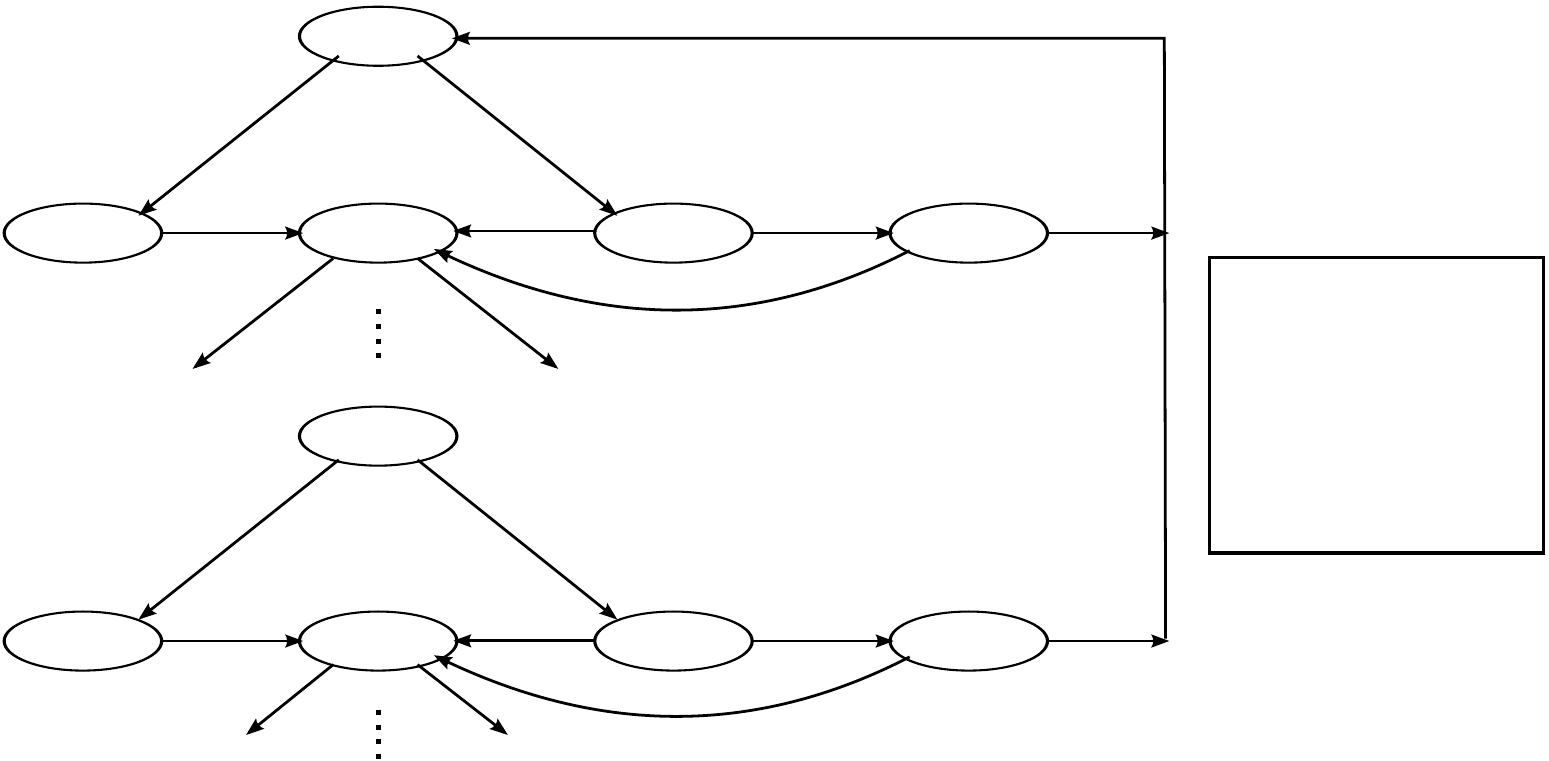}}%
    \put(0.794,0.185){\color[rgb]{0,0,0}\makebox(0,0)[lb]{\parbox[b][4em][c]{0.33\textwidth}{$I$: Idle \\ \\ $N$: Non-event \\ \\ $E$: Event \\ \\ $T$: Transmission}}}%
    \put(0.2275,0.4575){\color[rgb]{0,0,0}\makebox(0,0)[lb]{\smash{$I,0$}}}%
    \put(0.2275,0.3325){\color[rgb]{0,0,0}\makebox(0,0)[lb]{\smash{$I,1$}}}%
    \put(0.2275,0.07){\color[rgb]{0,0,0}\makebox(0,0)[lb]{\smash{$I,d$}}}%
    \put(0.035,0.3325){\color[rgb]{0,0,0}\makebox(0,0)[lb]{\smash{$N,1$}}}%
    \put(0.414,0.3325){\color[rgb]{0,0,0}\makebox(0,0)[lb]{\smash{$E,1$}}}%
    \put(0.606,0.3325){\color[rgb]{0,0,0}\makebox(0,0)[lb]{\smash{$T,1$}}}%
    \put(0.035,0.0699){\color[rgb]{0,0,0}\makebox(0,0)[lb]{\smash{$N,d$}}}%
    \put(0.414,0.0699){\color[rgb]{0,0,0}\makebox(0,0)[lb]{\smash{$E,d$}}}%
    \put(0.606,0.0699){\color[rgb]{0,0,0}\makebox(0,0)[lb]{\smash{$T,d$}}}%
    \put(0.1425,0.405){\color[rgb]{0,0,0}\rotatebox{38.00000034}{\makebox(0,0)[lb]{\smash{$q_{\gamma,1}$}}}}%
    \put(0.135,0.14){\color[rgb]{0,0,0}\rotatebox{38.00000034}{\makebox(0,0)[lb]{\smash{$q_{\gamma,d}$}}}}%
    \put(0.324,0.425){\color[rgb]{0,0,0}\rotatebox{-39.00000001}{\makebox(0,0)[lb]{\smash{$p_{\gamma,1}$}}}}%
    \put(0.324,0.1675){\color[rgb]{0,0,0}\rotatebox{-39.00000001}{\makebox(0,0)[lb]{\smash{$p_{\gamma,d}$}}}}%
    \put(0.3275,0.35){\color[rgb]{0,0,0}\makebox(0,0)[lb]{\smash{$q_{\alpha}$}}}%
    \put(0.51,0.35){\color[rgb]{0,0,0}\makebox(0,0)[lb]{\smash{$p_{\alpha}$}}}%
    \put(0.15,0.35){\color[rgb]{0,0,0}\makebox(0,0)[lb]{\smash{$1$}}}%
    \put(0.695,0.35){\color[rgb]{0,0,0}\makebox(0,0)[lb]{\smash{$q$}}}%
    \put(0.4275,0.30){\color[rgb]{0,0,0}\makebox(0,0)[lb]{\smash{$p$}}}%
    \put(0.3275,0.0875){\color[rgb]{0,0,0}\makebox(0,0)[lb]{\smash{$q_{\alpha}$}}}%
    \put(0.51,0.0875){\color[rgb]{0,0,0}\makebox(0,0)[lb]{\smash{$p_{\alpha}$}}}%
    \put(0.15,0.0875){\color[rgb]{0,0,0}\makebox(0,0)[lb]{\smash{$1$}}}%
    \put(0.695,0.0875){\color[rgb]{0,0,0}\makebox(0,0)[lb]{\smash{$q$}}}%
    \put(0.4275,0.04){\color[rgb]{0,0,0}\makebox(0,0)[lb]{\smash{$p$}}}%
    \put(0.21,0.1975){\color[rgb]{0,0,0}\makebox(0,0)[lb]{\resizebox{1.25cm}{!}{$I,d-1$}}}%
  \end{picture}%
\endgroup

%% file: figuresTR/Flowchart.eps_tex

\begingroup
  \makeatletter
  \providecommand\color[2][]{%
    \errmessage{(Inkscape) Color is used for the text in Inkscape, but the package 'color.sty' is not loaded}
    \renewcommand\color[2][]{}%
  }
  \providecommand\transparent[1]{%
    \errmessage{(Inkscape) Transparency is used (non-zero) for the text in Inkscape, but the package 'transparent.sty' is not loaded}
    \renewcommand\transparent[1]{}%
  }
  \providecommand\rotatebox[2]{#2}
  \ifx\svgwidth\undefined
    \setlength{\unitlength}{227.73857422pt}
  \else
    \setlength{\unitlength}{\svgwidth}
  \fi
  \global\let\svgwidth\undefined
  \makeatother
  \begin{picture}(1,0.75477055)%
    \put(0,0){\includegraphics[width=\unitlength]{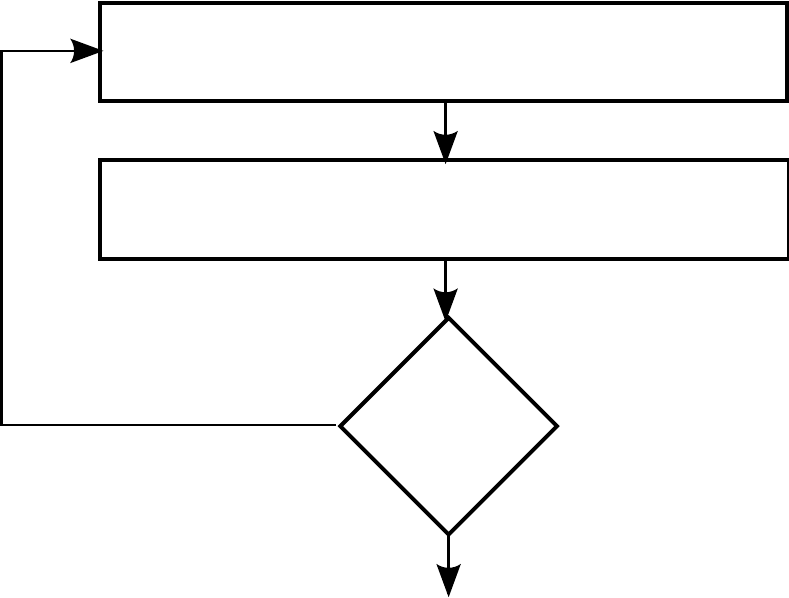}}%
    \put(0.4,0.67){\color[rgb]{0,0,0}\makebox(0,0)[lb]{\smash{Select $\{p_{_{\alpha,r}}\}$, $\{p_{_{\gamma,d}}\}$}}}%
    \put(0.4,0.47){\color[rgb]{0,0,0}\makebox(0,0)[lb]{\smash{Compute $p$, $\{\pi_{_{(I,d)}}\}$}}}%
    \put(0.5,0.195){\color[rgb]{0,0,0}\makebox(0,0)[lb]{\smash{Stable?}}}%
    \put(0.605,0.029559){\color[rgb]{0,0,0}\makebox(0,0)[lb]{\smash{Yes}}}%
    \put(0.27558881,0.25360354){\color[rgb]{0,0,0}\makebox(0,0)[lb]{\smash{No}}}%
  \end{picture}%
\endgroup

%% file: figuresTR/SymmRearrangement.eps_tex

\begingroup
  \makeatletter
  \providecommand\color[2][]{%
    \errmessage{(Inkscape) Color is used for the text in Inkscape, but the package 'color.sty' is not loaded}
    \renewcommand\color[2][]{}%
  }
  \providecommand\transparent[1]{%
    \errmessage{(Inkscape) Transparency is used (non-zero) for the text in Inkscape, but the package 'transparent.sty' is not loaded}
    \renewcommand\transparent[1]{}%
  }
  \providecommand\rotatebox[2]{#2}
  \ifx\svgwidth\undefined
    \setlength{\unitlength}{601.84106445pt}
  \else
    \setlength{\unitlength}{\svgwidth}
  \fi
  \global\let\svgwidth\undefined
  \makeatother
  \begin{picture}(1,0.19099433)%
    \put(0,0){\includegraphics[width=\unitlength]{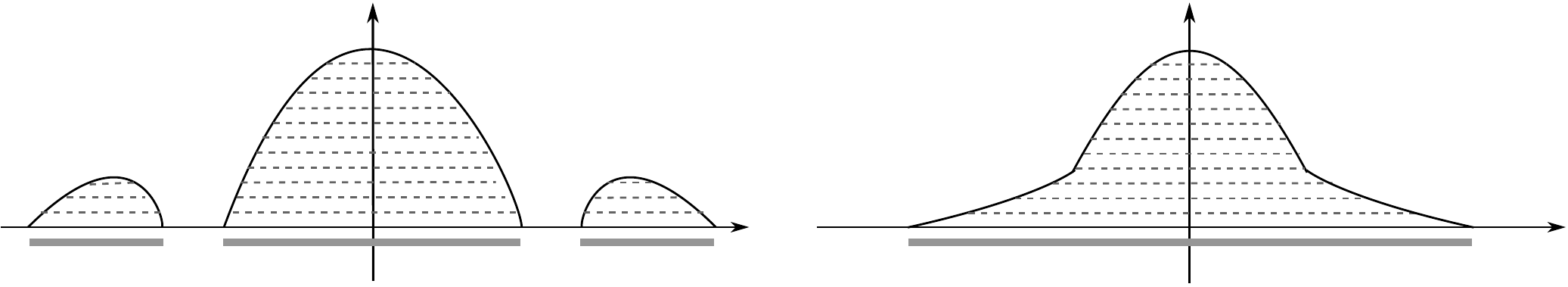}}%
    \put(0.46607834,0.01330685){\color[rgb]{0,0,0}\makebox(0,0)[lb]{\smash{$x$}}}%
    \put(0.98330031,0.01244394){\color[rgb]{0,0,0}\makebox(0,0)[lb]{\smash{$x$}}}%
    \put(0.71006759,0.17249807){\color[rgb]{0,0,0}\makebox(0,0)[lb]{\smash{$\phi_x$}}}%
    \put(0.19478081,0.1731515){\color[rgb]{0,0,0}\makebox(0,0)[lb]{\smash{$\phi_x$}}}%
    \put(0.10,-0.01){\color[rgb]{0,0,0}\makebox(0,0)[lb]{\smash{$G \subset \mathbb{R}$}}}%
    \put(0.625,-0.01){\color[rgb]{0,0,0}\makebox(0,0)[lb]{\smash{$G^\sigma \subset \mathbb{R}$}}}%
  \end{picture}%
\endgroup

%% file: figuresTR/LossyNWupperbound.eps_tex

\begingroup
  \makeatletter
  \providecommand\color[2][]{%
    \errmessage{(Inkscape) Color is used for the text in Inkscape, but the package 'color.sty' is not loaded}
    \renewcommand\color[2][]{}%
  }
  \providecommand\transparent[1]{%
    \errmessage{(Inkscape) Transparency is used (non-zero) for the text in Inkscape, but the package 'transparent.sty' is not loaded}
    \renewcommand\transparent[1]{}%
  }
  \providecommand\rotatebox[2]{#2}
  \ifx\svgwidth\undefined
    \setlength{\unitlength}{346.03649902pt}
  \else
    \setlength{\unitlength}{\svgwidth}
  \fi
  \global\let\svgwidth\undefined
  \makeatother
  \begin{picture}(1,0.15389204)%
    \put(0,0){\includegraphics[width=\unitlength]{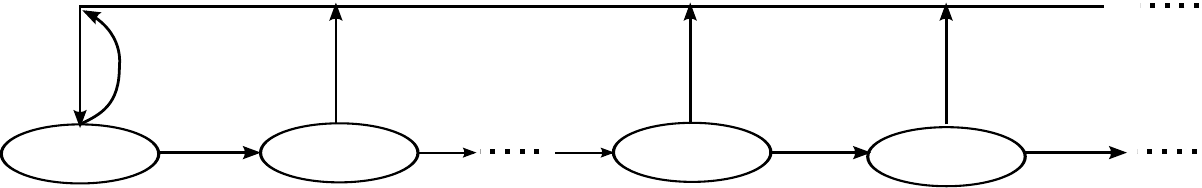}}%
    \put(0.04,0.01725){\color[rgb]{0,0,0}\makebox(0,0)[lb]{\smash{$I,0$}}}%
    \put(0.255,0.0175){\color[rgb]{0,0,0}\makebox(0,0)[lb]{\smash{$I,1$}}}%
    \put(0.52,0.0175){\color[rgb]{0,0,0}\makebox(0,0)[lb]{\smash{$I,d-1$}}}%
    \put(0.76,0.0165){\color[rgb]{0,0,0}\makebox(0,0)[lb]{\smash{$I,d$}}}%
    \put(0.15,0.04){\color[rgb]{0,0,0}\makebox(0,0)[lb]{\smash{$p_{l,1}$}}}%
    \put(0.65,0.04){\color[rgb]{0,0,0}\makebox(0,0)[lb]{\smash{$p_{l,d}$}}}%
    \put(0.107,0.1045){\color[rgb]{0,0,0}\makebox(0,0)[lb]{\smash{$q_{l,1}$}}}%
    \put(0.288,0.1045){\color[rgb]{0,0,0}\makebox(0,0)[lb]{\smash{$q_{l,2}$}}}%
    \put(0.586,0.1045){\color[rgb]{0,0,0}\makebox(0,0)[lb]{\smash{$q_{l,d}$}}}%
    \put(0.795,0.1045){\color[rgb]{0,0,0}\makebox(0,0)[lb]{\smash{$q_{l,d+1}$}}}%
  \end{picture}%
\endgroup